\documentclass[11pt,letterpaper]{article}
\usepackage{graphicx}
\usepackage{float}
\usepackage{color}
\usepackage{paralist}

\usepackage{epsfig}
\usepackage{amsfonts}
\usepackage{amssymb}
\usepackage{amstext}
\usepackage{amsmath}
\usepackage{xspace}
\usepackage{theorem}
\usepackage{mathrsfs}
\usepackage{tabularx}

\usepackage{enumerate}
\usepackage{url}

\usepackage{fullpage}
\usepackage[boxed,lined]{algorithm2e}
\usepackage{authblk}

\numberwithin{equation}{section}

\newenvironment{proof}{{\bf Proof:  }}{\hfill\rule{2mm}{2mm}}

\newtheorem{proposition}{Proposition}[section]
\newtheorem{definition}{Definition}[section]
\newtheorem{remark}{Remark}[section]

\newtheorem{theorem}{Theorem}[section]
\newtheorem{lemma}{Lemma}[section]

\newtheorem{fact}{Fact}[section]

\newenvironment{theorem*}{{\bf Theorem}}{}

{\theorembodyfont{\rmfamily}}

\newcommand{\R}{\ensuremath{\mathbb R}}

\newcommand{\N}{\ensuremath{\mathbb N}}

\newcommand{\A}{\ensuremath{\mathcal A}}

\newcommand{\D}{\ensuremath{\mathcal D}}

\newcommand{\calP}{\ensuremath{\mathcal P}}
\newcommand{\calR}{\ensuremath{\mathcal R}}
\newcommand{\scrA}{\ensuremath{\mathscr A}}

\newcommand{\expct}[1]{\ensuremath{\text{{\bf E}$\left[#1\right]$}}}
\newcommand{\expctsub}[2]{\ensuremath{\text{{\bf E}$_{#1}\left[#2\right]$}}}

\newcommand{\CP}{\ensuremath{\mathsf{CP}}}
\newcommand{\CD}{\ensuremath{\mathsf{CD}}}
\newcommand{\LP}{\ensuremath{\mathsf{LP}}}

\newcommand{\txts}{\textstyle}

\newcommand{\hubert}[1]{{\footnotesize\color{red}[Hubert: #1]}}

\newcommand{\ignore}[1]{}

\begin{document}

\title{A Primal-Dual Continuous LP Method on the Multi-choice Multi-best Secretary Problem}
\author{T-H. Hubert Chan}
\author{Fei Chen}
\date{}
\affil{The University of Hong Kong}
\clearpage\maketitle

\begin{abstract}
The $J$-choice $K$-best secretary problem, also known as the $(J,K)$-secretary problem, is a generalization of the classical secretary problem. An algorithm for the $(J,K)$-secretary problem is allowed to make $J$ choices and the payoff to be maximized is the expected number of items chosen among the $K$ best items. 

Previous works analyzed the case when the total number $n$ of items is finite,
and considered what happens when $n$ grows.  However, for general $J$ and $K$,
the optimal solution for finite $n$ is difficult to analyze.  Instead,
we prove a formal connection between the finite model and the infinite model,
where there are countably infinite number of items, each attached with a random arrival time drawn independently and uniformly from $[0,1]$.

We use primal-dual continuous linear programming techniques to analyze a class of infinite algorithms, which are general enough to capture the asymptotic behavior of the finite model with large number of items.
Our techniques allow us to prove that the optimal solution can be achieved by
the $(J,K)$-Threshold Algorithm, which has a nice ``rational description'' for 
the case $K = 1$.

%

\end{abstract}

\thispagestyle{empty}

\newpage

\pagestyle{plain}
\setcounter{page}{1}

\section{Introduction} \label{sec:intro}

The classical secretary problem has been popularized
in the 1950s, and since then various variants and solutions
for the problem have been studied. Freeman~\cite{Freeman83}
and Ferguson~\cite{Ferguson89} both have written survey
papers, which describe the history of the problem and contain
many references.

The most recent related paper to our work
is by Buchbinder, Jain and Singh~\cite{Buchbinder09},
who considered the
$J$-choice $K$-best secretary problem (also known as the $(J,K)$-secretary problem);
the $(J,1)$-case is sometimes simply referred to as the $J$-case.
The input to the problem is $n$ items (whose merits
are given by a total ordering) that arrive in a uniformly random permutation.
An \emph{algorithm} can observe the relative
merit of items arrived so far, and must decide irrevocably
if an item is chosen (or selected) when it arrives.
 In the $(J,K)$-case, the algorithm is allowed to make $J$ choices
 and the objective \emph{payoff} to be maximized is the expected number of items chosen
 among the best $K$ items, where expectation is over the random arrival permutation.
 The simple $(1,1)$-case is the classical secretary problem.  As observed
 in~\cite{Buchbinder09}, the $(K,K)$-case is equivalent to a variant
 considered by Kleinberg~\cite{Kleinberg05}.

For finite $n$ items, they gave a linear programming formulation,
which completely characterizes the $(J,K)$-secretary problem (and also other variants).
However, the limiting behavior for large $n$ is often the interesting
aspect of the problem.  For instance, the asymptotic optimal payoff for
the $(1,1)$-case is $\frac{1}{e}$ (where $e$ is the natural number), which is achieved by the simple algorithm
of discarding the first $\frac{n}{e}$ items, and after that choosing the first
\emph{potential}, where an item is a potential
if it is the best item arrived so far.
The authors in~\cite{Buchbinder09} were able to extend this method to
the $(2,1)$-case and showed that the optimal payoff
$\frac{1}{e} + \frac{1}{e^{1.5}}$ can be achieved by a similar
method involving the $\frac{n}{e}$-th and $\frac{n}{e^{1.5}}$-th items.
However, for more general cases, they did not give a way to 
analyze the asymptotic behavior,
nor how to derive a ``simple'' algorithm from an optimal LP solution.  For instance,
they have only claimed that the asymptotic optimal payoff for the $(1,2)$-case is about 0.572284
(which we later show is actually about 0.573567 whose value
we also verify by the finite LP for large $n$).

Other input models have been considered to analyze the problem for large $n$.
For instance, Bruss introduced the \emph{continuous model}~\cite{Bruss84}, in which
instead of a random arrival permutation, each item picks an arrival
time independently and uniformly at random from $[0,1]$.  
This formulation allows a \emph{threshold} algorithm, which does not need
to know $n$ in advance: the algorithm discards all items arriving before time $\frac{1}{e}$
and after time $\frac{1}{e}$ chooses the first potential.
However, this input model
is equivalent to the one before in terms of the optimal payoff,
and hence the asymptotic optimal payoff $\frac{1}{e}$ is approached only when
$n$ tends to infinity; for finite $n$, the optimal payoff is always
strictly larger than $\frac{1}{e}$. Would it not be neat to have
an abstract ``infinite'' instance of the problem whose optimal
payoff is the asymptotic optimal payoff for large $n$?

Immorlica et al.~\cite{Immorlica06} extended the continuous
model to the \emph{infinite model} where the number of items is countably infinite.
They considered multiple employers competing for the best item under the
infinite model, but there was no formal treatment for the connection
with the finite case.

\noindent \textbf{Our Contribution and Results.} We use the infinite model as
a tool to analyze the $(J,K)$-secretary problem for large $n$.  In particular,
we have made the following technical contributions.

\noindent (1) We give a formal treatment of the infinite model
and define a special class $\scrA$ of \emph{piecewise continuous} 
infinite algorithms, which
include the $(J,K)$-\emph{Threshold Algorithm}. For the $(J,1)$-case,
the algorithm is characterized by $J$ thresholds $0 < t_J \leq \cdots \leq t_1 \leq 1$
and items are selected according to the simple rule:
\begin{compactitem}
\item[(i)] At every threshold time $t_j$, 1 additional quota for selection is made
available to the algorithm; the initial quota is 0.
\item[(ii)] When an arriving item is a potential (i.e., the best among all arrived items)
and the number of available quotas is non-zero, then the algorithm selects the item and uses one available quota.
\end{compactitem}
We define the general $(J,K)$-Threshold
Algorithm in Section~\ref{sec:pre}.

\noindent (2) We show that the $(J,K)$-secretary problem restricted to our class $\scrA$
of infinite algorithms has a \emph{continuous linear programming} formulation (see Tyndall~\cite{Tyndall65} and Levinson~\cite{Levinson66}),
which extends the LP formulation for finite $n$ in~\cite{Buchbinder09}.  Moreover,
we show that the optimal payoff of the finite LP approaches that of the continuous LP
for large $n$.  
Furthermore,
if an optimal infinite algorithm also satisfies some
\emph{monotone} property (which we show
indeed is the case for Threshold Algorithms),
then the algorithm
can be applied
to any finite case with the same payoff.  Hence, we have established
a formal connection between the infinite model (under algorithm class $\scrA$) and the finite model.

\noindent (3) By considering duality and complementary slackness properties
of continuous LP,  we give a \emph{clean} primal-dual method
to find the optimal thresholds for the $(J,K)$-Threshold
Algorithm, and at the same time prove that the algorithm
is optimal by exhibiting a feasible dual that satisfies
complementary slackness.  In particular, we discover that the optimal strategy
for the $(J,1)$-case has a very nice structure, that has a representation using
rational numbers.  We believe it would be too tedious to directly
analyze the limiting behavior of the finite model for reaching the same
conclusion.

%
%


\begin{theorem}[Main Theorem]
\label{th:main0}
There is a procedure to find appropriate thresholds for which
the $(J,K)$-Threshold Algorithm (defined in Section~\ref{sec:pre}
together with thresholds $(\tau_{j,k})_{j \in [J], k \in [K]}$)
is optimal in the infinite model (under algorithm class $\scrA$), and the algorithm can be applied
to the finite model with $n$ items to achieve optimal asymptotic payoff for large $n$. The optimal payoff is $J-\sum_{j=1}^{J} (1-\tau_{j,1})^K$.
\end{theorem}

For $K=1$ with thresholds denoted by $t_j = \tau_{j,1}$ for $j\in[J]$, 
the optimal solution has a particularly nice structure (see Table~\ref{table:J}).
Moreover, there is a method to generate the $J$ optimal thresholds in $O(J^3)$ time, which we explicitly describe in Section~\ref{sec:generate}.


\begin{theorem}[$(J,1)$-Secretary Problem] \label{th:main1}
There is a procedure to construct an increasing sequence $\{\theta_j\}_{j \geq 1}$
of rational numbers such that for any $J \geq 1$,
the optimal $J$-Threshold Algorithm
uses thresholds $\{t_j := \frac{1}{e^{\theta_j}} |\, 1 \leq j \leq J\}$ (that
can be computed in $O(J^3)$ time)
and has payoff $\sum_{j=1}^J t_j$.
\end{theorem}

\begin{table}[H]
\centering
\begin{tabular}{|c|c|c|}
\hline
J & Payoff & $\theta_J$ \\
\hline
1 & 0.367879 & $1$ \\
\hline
2 & 0.591010 & $\frac{3}{2} = 1.5$ \\
\hline
3 & 0.732103 & $\frac{47}{24} \approx 1.958333$ \\
\hline
4 & 0.823121 & $\frac{2761}{1152} \approx 2.396701$ \\
\hline
5 & 0.882550 & $\frac{4162637}{1474560} \approx 2.822969$ \\
\hline
6 & 0.921675 & $\frac{380537052235603}{117413668454400} \approx  3.240994$ \\
\hline
7 & 0.947588 & $\frac{705040594914523588948186792543}{193003573558876719588311040000} \approx 3.652992$ \\
\hline
8 & 0.964831 & $\frac{302500210177484374840641189918370275991590974715547528765249}{74500758812993473612938854416966977838930799571763200000000} \approx 4.060364$ \\
\hline
\end{tabular}
\caption{Optimal Payoffs for $J$-secretary problem}
\label{table:J}
\end{table}

For general $K \geq 2$, 
the optimal solution does not have a nice structure, but
the continuous LP still allows us to compute the exact solution for some cases.
To describe the optimal solution, we use part of the principal
branch of the \emph{Lambert W function}~\cite{Lambert1758} $W : [-\frac{1}{e}, 0] \rightarrow [-1,0]$,
where $z = W(z)e^{W(z)}$ for all $z \in [-\frac{1}{e}, 0]$.

\begin{theorem}[$(J,2)$-Secretary Problem for $J =1,2$] \label{th:12_22}
Define the thresholds:
$\tau_{1,2} = \frac{2}{3}$; $\tau_{1,1} = - W(-\frac{2}{3e}) \approx 0.346982$;

\noindent $\tau_{2,2} \approx 0.517291$ is the solution of: $x \ln x + \ln x - (2+3\ln\frac{2}{3}) x + 1 - \ln\frac{2}{3}=0$;
$\tau_{2,1} = - W(-e^{-c/2}) \approx 0.227788$, where $c := -(\ln \tau_{1,1})^2 + 2 \ln\frac{2}{3}\ln \tau_{1,1} + (\ln \tau_{2,2})^2 - 2\ln\frac{2}{3}\ln \tau_{2,2} - 2 \tau_{2,2} + 4 - 2\ln\frac{2}{3}$.

Then, these thresholds can be used to achieve the following optimal payoffs:
\begin{compactitem}
\item[(a)] $(1,2)$-case: $2 \tau_{1,1} - \tau_{1,1}^2 \approx 0.573567$.

\item[(b)] $(2,2)$-case: $(2 \tau_{1,1} - \tau_{1,1}^2) + (2 \tau_{2,1} - \tau_{2,1}^2) \approx 0.977256$.
\end{compactitem}
\end{theorem}

We give a complete proof for Theorem~\ref{th:main1} in Section~\ref{sec:optalg}, and give the proofs for Theorem~\ref{th:main0} and Theorem~\ref{th:12_22} in Section~\ref{sec:k-best}.

\section{Preliminaries} \label{sec:pre}

We use the infinite model as a tool to analyze the 
secretary problem when the number $n$ of items is large.
We shall describe the properties of our ``infinite'' algorithms,
which can still be applied to finite instances
to obtain conventional algorithms.
We consider countably infinite number of items, whose ranks
are indexed by the set $\N$ of positive integers,
where lower rank means better merit.  Hence, the
item with rank 1 is the best item.  The arrival
time of each item is a real number drawn
independently and uniformly at random from $[0,1]$ 
(where the probability that two items
arrive at the same time is 0);
the (random) function $\rho: \N \rightarrow [0,1]$ gives
the arrival time of each item, 
where $\rho(i)$ is the arrival time of the item with rank $i$.
For a positive integer $m$, we denote $[m] := \{1,2,\ldots, m\}$.

\noindent \textbf{Input Sample Space.}  An algorithm can
observe the arrival times $\Sigma$ of items
and their relative merit, which can be given
by a total ordering $\prec$ on $\Sigma$.
Given $\rho : \N \rightarrow [0,1]$, we have the set
$\Sigma_\rho := \{\rho(i)|\, i\in\N \}$,
and a total ordering $\prec_\rho$ on $\Sigma_\rho$ defined by
$\rho(i) \prec_\rho \rho(j)$ if and only if $i<j$.
The sample space is 
$\Omega := \{(\Sigma_\rho, \prec_\rho) |\, \rho : \N \rightarrow [0,1]\}$,
with a probability distribution induced by
the randomness of $\rho$; we say each $\omega = (\Sigma, \prec) \in \Omega$
is an \emph{arrival sequence}. We sometimes use time $t \in \Sigma$
to mean the item arriving at time $t$, for instance
we might say ``the algorithm selects $t \in \Sigma$.''

\begin{fact}[Every Non-Zero Interval Contains Infinite Number of Items] \label{fact:fill}
For every interval $\mathcal{I} \subseteq [0,1]$ with non-zero length, the probability that there exist infinitely many items arriving in $\mathcal{I}$ is 1.
\end{fact}

\noindent \textbf{Infinite Algorithm.} When an item arrives,
an algorithm must decide immediately whether
to select that item.  Moreover, an algorithm
does not know the absolute ranks of the items, 
but can observe only the relative merit of the items seen so far.
These properties are captured
for our infinite algorithms as follows.  Given $\omega = (\Sigma, \prec) \in \Omega$,
and $t \in [0,1]$, let $\Sigma^{[t]} := \{x \in \Sigma|\, x \leq t\}$
be the arrival times up to time $t$,
with the total ordering inherited from $\prec$, which
strictly speaking can be denoted by $\prec_{\Sigma^{[t]}}$.
However, for notational convenience, we write 
$\omega^{[t]} = (\Sigma^{[t]}, \prec)$, dropping the subscript for $\prec$.
We denote $\Omega^{[t]} := \{\omega^{[t]} |\, \omega \in \Omega\}$.
Also we denote $\omega^{(t)} := (\Sigma^{[t]}\setminus \{t\}, \prec)$ and $\Omega^{(t)} := \{\omega^{(t)} | \omega\in\Omega \}$.

An infinite algorithm $\A$ is an ensemble of functions
$\{A^{[t]}: \Omega^{[t]} \rightarrow \{0,1\} |\, t \in [0,1]\}$,
where for time $t$, the value 1 means an item is chosen at time $t$ and 0 otherwise.  Observe that if no item arrives at time $t$, an algorithm
cannot select an item at that time; this means
for all $\omega = (\Sigma, \prec) \in \Omega$, if $t \notin \Sigma$,
then $\A^{[t]}(\omega^{[t]}) = 0$.
Any $J$-choice algorithm $\A$ must also satisfy that
for any $\omega \in \Omega$,
there can be at most $J$ values of $t$ such that 
$\A^{[t]}(\omega^{[t]}) = 1$.  
As we see later, it will be helpful to imagine that 
there are $J$ quotas $Q_J, Q_{J-1}, \ldots, Q_1$
available for making selections, where a quota with larger
index is used first. For instance $Q_J$ is used first and $Q_1$ is used last;
we use this ``reverse'' order to be compatible
with the description in Theorem~\ref{th:main1}.

We can assume an infinite algorithm can be applied
when the number $n$ of items is finite,
because this is equivalent to the arrival sequence
in which all items with ranks at most $n$ arrive before
all items with ranks larger than $n$ (although this happens
with zero probability).

An algorithm $\A: \Omega \rightarrow \{0,1, \ldots, K\}$ can
also be interpreted as a function, which returns the
number of items selected among the $K$ best items.  Since
we wish to maximize the expected payoff of an algorithm
where randomness comes from $\Omega$, 
we can consider
only deterministic algorithms without loss of generality.

\begin{definition} [Outcome and Payoff]
Let $\A$ be an (infinite) algorithm. For $\omega \in \Omega$, 
the \emph{outcome} $\A(\omega)$ is the number of items selected
among the $K$ best items. The \emph{payoff} of $\A$ is defined as $P(\A) := \expctsub{\omega}{\A(\omega)}$.
\end{definition}

The reason we consider the infinite model is that
for any $0 < t \leq 1$, the sample space $\Omega^{(t)}$ observed
before $t$
has the same structure as $\Omega$ in the sense described in the
following Proposition~\ref{prop:iso_omega_t}.
This allows us to analyze the recursive behavior of any infinite algorithm.

\begin{proposition}[Isomorphism between $\Omega^{(t)}$ and $\Omega$] 
\label{prop:iso_omega_t}
For any $0 < t \leq 1$, the sample space $\Omega^{(t)}$ (with
distribution inherited from $\Omega$) rescaled to $[0,1]$ (by dividing
each arrival time by $t$) has the same distribution as $\Omega$.
\end{proposition}

\begin{proof}
Recall that the probability distribution over $\Omega$ is induced by the randomness of all the infinite arrival times, each of which is a random number drawn independently and uniformly from $[0,1]$. Similarly, the probability distribution over $\Omega^{(t)}$ is induced by the randomness of arrival times before $t$, the number of which is infinite by Fact~\ref{fact:fill}. Moreover, each arrival time in $[0,t)$ is drawn independently and uniformly from $[0,t)$, which after rescaling is independently and uniformly distributed in $[0,1]$.
\end{proof}

\begin{definition} [Potential]
Let $\omega = (\Sigma, \prec) \in \Omega$
be an arrival sequence.  For $k \geq 1$, each $t \in \Sigma$,
the item arriving at $t$ or time $t$ is a $k$-\emph{potential}
if the item is the $k$-th best item (with respect to
$\prec$) among those arrived by $t$; we sometimes refer to a $1$-potential
simply as a potential.  We say an item is a $k_{\geq}$-potential
(pronounced as ``at least $k$-potential'')  if it is
a $k'$-potential for some $k' \leq k$.
%
\end{definition}

\begin{proposition}[Distribution of Potentials] \label{prop:finite_can}
For every $k \geq 1$ and $t > 0$, with probability 1, the following conditions hold.
\begin{compactitem}
\item[1.] There exists a potential in $[0,t)$.
\item[2.] There are finitely many $k$-potentials in $[t, 1]$.
\end{compactitem}
\end{proposition}

\begin{proof}
From Fact~\ref{fact:fill},
with probability 1, there exists an item arriving in $[0,t)$.
This implies that there exists $i \in \N$ such that
the item with rank $i$ arrives at $\rho(i) \in [0,t)$.
If $\rho(i)$ is not a potential, then a non-empty subset $S$ of items
with ranks in $\{1,\ldots,i-1\}$ must have arrived before $\rho(i)$. Since $S$ is finite, the item among them with smallest arrival time is a potential. Thus, there is a potential in $[0,t)$ with probability 1. 

Similarly, from Fact~\ref{fact:fill}, there exist
$k$ items in $[0,t)$ and let $r$ be the maximum
rank among those $k$.  Every $k$-potential in $[t, 1]$ must have a
rank in $\{1,\ldots,r-1\}$; that is, there are finite number of  $k$-potentials in $[t, 1]$ with probability 1.
\end{proof}

\ignore{
Observe that since there are finite potentials in $[\epsilon,1]$ for every $\epsilon>0$, the best item (the one with rank $1$) 
is the last potential arriving in $[0,1]$.
Since we care only about whether an algorithm selects
among the best $K$ items, we can assume without loss
of generality that an algorithm selects only $K_{\geq}$-potentials.}
Generalizing the simple $(J,1)$-Threshold Algorithm in the introduction,
we define the general version for the $(J,K)$-case.

\noindent \textbf{$(J,K)$-Threshold Algorithm.}
The algorithm takes $JK$ thresholds $(\tau_{j,k})_{j\in[J],k\in[K]}$  such that
(i) for all $k \in [K]$, $0< \tau_{J,k} \leq \tau_{J-1,k} \leq \cdots \leq \tau_{1,k} \leq 1$; and (ii) for all $j \in [J]$, $0 <\tau_{j,1} \leq \tau_{j,2}\leq \cdots \leq \tau_{j,K} \leq 1$.

Items are selected according to the following rules.
\begin{compactitem}
\item[(a)] For each $j \in [J]$, at time $\tau_{j,1}$ a quota $Q_j$ is
made available to the algorithm.
\item[(b)] For each $j \in [J]$ and $k \in [K]$, after time $\tau_{j,k}$,
the algorithm can select a $k_\geq$-potential by using up an available
quota $Q_{j'}$, for some $j' \geq j$. (We require
that the available quota $Q_{j'}$ with the largest $j'$ is used.)
Selection is done greedily, i.e., the algorithm will select an arriving item
whenever it is possible according to the above rule.
\end{compactitem}

\begin{remark}
We can imagine that each quota $Q_j$ has different maturity
times. For instance, at time $\tau_{j,1}$, the quota can only be
used for selecting 1-potential.  Hence, condition (ii)
means that there are $K$ progressive maturity times,
where after time $\tau_{j,k}$, quota $Q_j$ can be used
for selecting $k_{\geq}$-potentials.  Condition (i) means
that quotas with larger indices mature to the next stage earlier.
Note that the first quota is released
at time $\tau_{J,1} > 0$. By Proposition~\ref{prop:finite_can},
with probability 1, there are only a finite number of $K_\geq$-potentials
arriving after $\tau_{J,1}$, and hence the algorithm is well-defined.
We could still give a formal treatment if for all $\epsilon$, there exists a potential
arriving in $(\tau_{J,1}, \tau_{J,1} + \epsilon)$; however,
we omit the details as this case happens with probability 0.
\end{remark}


\noindent\textbf{Piecewise Continuous.} Given an algorithm $\A$, $j \in [J]$
and $k \in [K]$, define the function $p^\A_{j|k}:[0,1]\rightarrow[0,1]$ such that $p^\A_{j|k}(x)$ is the probability that $\A$ selects time $x$ using quota $Q_j$ given that $x$ is a $k$-potential. Let $p^\A=(p^\A_{j|k})_{j \in [J], k \in [K]}$ be the collection of functions for $\A$. We say 
$\A$ is \emph{piecewise continuous} if every $p^\A_{j|k}$ is piecewise continuous.  We denote by $\scrA$ the class of
piecewise continuous algorithms; as we shall see, this class of
algorithms is general enough to capture the asymptotic behavior
for finite models with large $n$ number of items.

\begin{proposition}
The $(J,K)$-Threshold Algorithm is piecewise continuous.
\end{proposition}

\begin{proof}
The argument is straightforward but the full proof is tedious.
As a special case, consider the first threshold $\tau_{J,1}$ and the function
$p_{J|1}(x)$ giving the probability
that a potential arriving at time $x$
is selected by using quota $Q_J$, which is 0 for $x < \tau_{J,1}$,
and there is a discontinuity at $\tau_{J,1}$. At time
$x > \tau_{J,1}$,
the probability $p_{J|1}(x)$ is the same as that for the event of
the best item before time $x$ arriving before $\tau_{J,1}$,
and so $p_{J|1}(x) = \frac{\tau_{J,1}}{x}$ for $x > \tau_{J,1}$.
Other $p_{j|k}$'s can be analyzed similarly.
\end{proof}


\ignore{
For $0 < t \leq 1$ and $\omega = (\Sigma, \prec) \in \Omega^{[t]}$,
we denote $\Gamma_K(\omega):= (\Sigma_K, \prec)$, where $\Sigma_K$ is the set of $K_{\geq}$-potentials in $\Sigma$ with respect
to the total ordering $\prec$.
We further define two properties of an algorithm.

\begin{definition} [Economical]
An infinite algorithm $\A$ is $K$-\emph{economical} if it selects
only $K_{\geq}$-potentials and its behavior only depends on the arrival times
of $K_{\geq}$-potentials seen so far. Formally, the algorithm
can be described by an ensemble of functions 
$\{\A^{[t]}: \Gamma_K(\Omega^{[t]}) \rightarrow \{0,1\} |\, t \in [0,1]\}$
such that for each $\omega \in \Omega$ and $K_{\geq}$-potential $t$ in $\omega$,
the choice made at $t$ is given by $\A^{[t]}(\Gamma_K(\omega^{[t]}))$.
%
%
\end{definition}
}

\noindent \textbf{Monotone.}
An algorithm $\A$ is \emph{monotone}
if for all $t \in [0,1]$, for any positive integer $n$, for any $\omega = (\Sigma, \prec) \in \Omega^{[t]}$ such that all items with ranks larger than $n$ are removed
to produce $\omega_n = (\Sigma_n, \prec)$, where the item at time $t$ is not removed, then it holds that $\A^{[t]}(\omega_n)$ is well-defined and is at least
$\A^{[t]}(\omega)$;
in other words, if in an arrival sequence the algorithm selects some $t$ with rank $i$, then if all items with ranks greater than $n$ (for some $n > i$) are removed, the algorithm must also select $t$.
The next proposition immediately follows.

\begin{proposition}[Monotone Algorithm Applicable to Finite Model]
\label{prop:monotone}
If a monotone infinite algorithm $\A$ has payoff $x$, then
for all positive integers $n$, the payoff of applying $\A$ to the finite model
with $n$ items is at least $x$.
\end{proposition}

\begin{proposition}
The $(J,K)$-Threshold Algorithm is monotone.
\end{proposition}

\begin{proof}
This follows from the observation that 
an item $x$ having rank $i$ is a $k$-potential
in some (infinite) arrival sequence \emph{iff}
it is a $k$-potential after all items with ranks greater than $n$ (for some $n > i$) are removed.
\end{proof}

\noindent \textbf{Main Approach.}  In Section~\ref{sec:lp},
we consider the continuous LP, which gives a connection
to the (finite) LP.  Suppose $\mathcal{P}^*$ is the optimal
payoff for the infinite model under algorithm class $\scrA$,
and for each positive integer $n$, $\mathcal{P}_n^*$ is the optimal payoff
for the finite model with $n$ items.
From Proposition~\ref{prop:p_cp} in Section~\ref{sec:lp},
we can conclude that $\limsup_{n \rightarrow \infty} \mathcal{P}_n^* \leq \CP^*$,
where $\CP^*$ is the optimal value of some continuous LP formulation
of the secretary problem.

In Section~\ref{sec:optalg}, we show that the optimal payoff
for the infinite algorithm can be achieved by the Threshold Algorithm,
which is monotone.  In particular,  we show that $\mathcal{P}^* = \CP^*$ and by Proposition~\ref{prop:monotone}
this implies that for all $n$, $\mathcal{P}^* \leq \mathcal{P}_n^*$.
Hence, it follows that $\lim_{n \rightarrow \infty} \mathcal{P}_n^* = \mathcal{P}^*$,
and the Threshold Algorithm achieves the asymptotic optimal payoff for the finite
model for large $n$, as stated in Theorem~\ref{th:main0}.



\ignore{
Our main result is that there exists a $K$-economical
and monotone algorithm that is optimal for the $(J,K)$-secretary problem. Moreover, the maximum payoff in the finite model with $n$ items tends to the maximum payoff in the infinite model as $n$ goes to infinity. 
\hubert{Result true for general $(J,K)$?}

\begin{theorem} [Existence of Economical and Monotone Optimal Algorithm] \label{th:opt}
There exists an optimal algorithm that is economical and monotone for the $J$-choice secretary problem in the infinite model.
\end{theorem}

\begin{theorem} [Equivalence of the Infinite and Finite Models] \label{th:equiv}
Let $\calP = \max_{\A\in\scrA} \{P(\A)\}$ be the maximum payoff for the $J$-choice secretary problem in the infinite model. Let $\calP^n$ be the maximum payoff for the $J$-choice secretary problem in the finite model with $n$ items. Then, $\lim_{n\rightarrow \infty} \calP^n = \calP$.
\end{theorem}
}

\section{The Continuous Linear Programming} \label{sec:lp}

For the finite model with random permutation, 
Buchbinder et al.~\cite{Buchbinder09} showed that there exists a 
linear programming $\LP_n(J,K)$ such that there is a one-to-one correspondence between an algorithm for the $(J,K)$-secretary problem with $n$ items and a feasible solution of the LP; the payoff of the algorithm is exactly the objective of $\LP_n(J,K)$. 
Therefore, the optimal value of the $\LP_n(J,K)$ gives the maximum payoff of the $(J,K)$-secretary problem with $n$ items.  
We rewrite their LP in a convenient form; recall
that the quotas are used in the order $Q_J, Q_{J-1}, \ldots, Q_1$.
The variable $z_{j|k}(i)$ represents that the probability
that the $i$-th item is selected using quota $Q_j$ given that it is
a $k$-potential.
\begin{align*}
\txts \LP_n(J,K) \qquad \max \qquad & v_n(z) = \txts\sum_{j=1}^J \sum_{k=1}^{K} \sum_{i=1}^n \frac{1}{n} \sum_{\ell=k}^{K} \frac{\binom{n-i}{\ell-k} \binom{i-1}{k-1}}{\binom{n-1}{\ell-1}} z_{j|k}(i) \\
\text{s.t.} \qquad & z_{j|k}(i) \leq \txts\sum_{m=1}^{i-1} \frac{1}{m} \sum_{\ell=1}^{K} [z_{(j+1)|\ell}(m) - z_{j|\ell}(m)], \\
& \qquad \qquad \forall i\in[n], k\in[K], 1 \leq j < J \\
 & z_{J|k}(i) \leq 1 - \txts\sum_{m=1}^{i-1} \frac{1}{m}\sum_{\ell=1}^{K} z_{J|\ell}(m), \qquad \forall i\in[n], k\in[K] \\
& z_{j|k}(i) \geq 0, \qquad \forall i\in[n], k\in[K], j\in[J].
\end{align*}

For the $(J,K)$-secretary problem in the infinite model, we construct a continuous linear programming such that every piecewise continuous algorithm
corresponds to a feasible solution, whose objective value is the payoff of the algorithm.
Hence, the optimal LP gives an upper bound for the maximum payoff $\calP^*$;
we later show that the Threshold Algorithm can achieve the optimal LP value.

For each $j\in [J]$ and $k\in [K]$, let $p_{j|k} (x)$ be a function of $x$ that is piecewise continuous in $[0,1]$. In the rest of this paper, we use  ``$\forall x$'' to denote ``for almost all $x$'',
which means for all but a measure zero set. Define $\CP(J,K)$ as follows.
\begin{align*}
\CP(J,K) \qquad \max \qquad & w(p) = \txts\sum_{j=1}^{J} \sum_{k=1}^{K} \int_{0}^{1} \left(\sum_{\ell=k}^{K} \binom{\ell-1}{k-1} (1-x)^{\ell-k} \right) x^{k-1} p_{j|k} (x) dx \\
\text{s.t.} \qquad & p_{j|k}(x) \leq \txts\int_{0}^{x} \frac{1}{y} \sum_{\ell=1}^{K} [p_{(j+1)|\ell}(y) - p_{j|\ell} (y)] d y, \\
& \qquad\qquad {\forall} x\in[0,1], k\in[K], 1 \leq j < J \\
& p_{J|k}(x) \leq 1 - \txts\int_{0}^{x} \frac{1}{y} \sum_{\ell=1}^{K} p_{J|\ell}(y) d y, \qquad {\forall} x\in[0,1], k\in[K] \\
& p_{j|k} (x) \geq 0, \qquad {\forall} x\in[0,1], k\in[K], j\in[J].
\end{align*}

Fix an algorithm $\A\in\scrA$. 
For each $x$ and $j$ and $k$, the events $E_x^j$, $Z_x^j$, $V_x^k$ and $W_x^k$ are defined as follows. Let $E_x^j$ be the event that time $x$ is selected using quota $Q_j$. 
Let $Z_x^j$ be the event that quota $Q_j$ has already been used before time $x$, i.e.,
all quotas $Q_{j'}$ for $j' \geq j$ have been used.
Let $V_x^k$ be the event that time $x$ is a $k$-potential. Let $W_x^k$ be the event that time $x$ is the $k$-th best item overall.
Note that $Z_x^{j}$ implies $Z_x^{j+1}$, and $Z_x^{j+1} \wedge \overline{Z_x^{j}}$ is the event that
quota $Q_j$ is the next quota available to be used at time $x$,
for $1 \leq j < J$.
Also observe that $E_x^j$ implies $Z_x^{j+1} \wedge \overline{Z_x^{j}}$. 


\begin{lemma}[Independence between Potential and Past History] \label{lemma:ind}
For $0< x \leq 1$, and positive integer $k$, the event $V_x^k$
that $x$ is a $k$-potential
is independent of the
arrival sequence observed before time $x$.  In particular, this implies that
for any $K > 1$, the event that $x$ is a
$K_\geq$-potential is also independent of the arrival sequence observed before time $x$.
\end{lemma}
\begin{proof}
By Proposition~\ref{prop:iso_omega_t} the arrival sequence observed before time $x$ can be generated by sampling a random arrival time for each integer in $\N$ independently and uniformly in $[0,x)$. We distinguish two cases: (1) without knowledge of $x$, this sequence is generated for all integers in $\N$; (2) given that $x$ is a $k$-potential for some $k\in[K]$, this sequence is generated for all integers in $\N\setminus\{k \}$. Since the total ordering on a sequence observed before $x$ is inherited from $\N$ and there is a bijection between $\N$ and $\N\setminus\{k \}$, the sequences generated in the two cases have the same distribution.
Hence, the event $V_x^k$ is independent of $\Omega^{(x)}$. Since the $V_x^k$'s for $k\in[K]$ are disjoint, the event that $x$ is a $K_\geq$-potential and $\Omega^{(x)}$ are independent.
\end{proof}

\ignore{
\begin{proof}

By Proposition~\ref{prop:iso_omega_t} the arrival sequence observed before time $x$ can be generated by sampling a random arrival time for each integer in $\N$ independently and uniformly in $[0,x)$. We distinguish two cases: (1) without knowledge of $x$, this sequence is generated for all integers in $\N$; (2) given that $x$ is a $k$-potential for some $k\in[K]$, this sequence is generated for all integers in $\N\setminus\{k \}$. Since the total ordering on a sequence observed before $x$ is inherited from $\N$ and there is a bijection between $\N$ and $\N\setminus\{k \}$, the sequences generated in the two cases have the same distribution.
Hence, the event $V_x^k$ is independent of $\Omega^{(x)}$. Since the $V_x^k$'s for $k\in[K]$ are disjoint, the event that $x$ is a $K_\geq$-potential and $\Omega^{(x)}$ are independent.
\end{proof}
}

\begin{lemma} \label{lemma:k-trick}
For all $j\in[J]$ and $x\in[0,1]$, we have $\Pr(Z_x^j) = \int_{0}^{x} \frac{1}{y}\sum_{\ell=1}^{K} p_{j|\ell}(y) d y$.
\end{lemma}

\begin{proof}
For $\ell\in[K]$, let $y_\ell$ be the arrival time of the $\ell$-th best item in $[0,x]$. Define $Y:=\max_{\ell\in[K]}\{y_\ell \}$. Then for each $y\in[0,x]$ we have $\Pr(Y\leq y) = \frac{y^K}{x^K}$. It follows that the probability density function of $Y$ is $f(y) = \frac{K y^{K-1}}{x^K}$. Also note that given $Y=y$, we have $\Pr(y_\ell=y) = \frac{1}{K}$ for all $\ell\in[K]$. It follows that $\Pr(E_y^j|Y=y) = \sum_{\ell=1}^{K} \Pr(E_y^j|V_y^\ell) \Pr(y_\ell = y) = \frac{1}{K}\sum_{\ell=1}^{K} p_{j|\ell}(y)$. 

There is no $K_\geq$-potential in $(y,x]$ and hence no item is selected. Thus $Z_x^j$ happens if and only if either $Z_y^j$ or $E_y^j$ (i.e. $Z_{y}^{j+1} \wedge \overline{Z_y^j} \wedge E_y^j$) happens.
By using arguments similar to the proof of Proposition~\ref{prop:iso_omega_t}, 
we can show that, whether the event $Y=y$ happens or not, the
distribution of sample space $\Omega^{(y)}$ of arrival time observed before time $y$ remains the same; in particular, the events $Z_y^j$ and $Y=y$ are independent.
Moreover the events $Z_y^j$ and $E_y^j$ are disjoint. By the law of total probability we have
\begin{align*}
\Pr(Z_x^j) & = \txts\int_{0}^{x} \Pr(Z_x^j | Y=y) f(y) d y \\
& = \txts\int_{0}^{x} [\Pr(Z_y^j|Y=y)+\Pr(E_y^j|Y=y)] \frac{Ky^{K-1}}{x^K} d y \\
& = \txts\frac{K}{x^K} \int_{0}^{x} [\Pr(Z_y^j) + \frac{1}{K}\sum_{\ell=1}^{K} p_{j|\ell}(y)] y^{K-1} d y.
\end{align*}
Fix $j$ and let $g(x) := \Pr(Z_x^j)$ be a function with respect to $x$. Taking derivatives on both sides of $x^K g(x) = K \int_{0}^{x} [g(y) + \frac{1}{K}\sum_{\ell=1}^{K} p_{j|\ell}(y)] y^{K-1} d y$ and using
piecewise continuity, we have $g'(x) = \frac{1}{x}\sum_{\ell=1}^{K} p_{j|\ell}(y) d y$
for almost all $x$. Then $g(x) = \int_{0}^{x} \frac{1}{y}\sum_{\ell=1}^{K} p_{j|\ell}(y) d y + c$ for some constant $c$. By definition $g(0)=0$ and thus $c=0$. Therefore we have
$\Pr(Z_x^j) = \int_{0}^{x} \frac{1}{y}\sum_{\ell=1}^{K} p_{j|\ell}(y) d y$.
\end{proof}

\begin{proposition} [Optimal Payoff At Most Optimal $\CP(J,K)$] \label{prop:cont_lp_k}
Let $\A \in \scrA$ be an algorithm for the $(J,K)$-secretary problem. Let $p=(p_{j|k})_{j\in[J],k\in[K]}$ be the functions such that for
 $j \in [J]$ and $k \in [K]$ and $x \in [0,1]$, the probability that time $x$ is selected by $\A$ using quota $Q_j$ given that time $x$ is a $k$-potential is $p_{j|k}(x)$. Then $p$ is a feasible solution of $\CP(J,K)$. Moreover, the payoff of $\A$ is exactly the objective
  $w(p)=\sum_{j=1}^{J} \sum_{k=1}^{K} \int_{0}^{1} \left(\sum_{\ell=k}^{K} \binom{\ell-1}{k-1} (1-x)^{\ell-k} \right) x^{k-1} p_{j|k} (x) dx$.
\end{proposition}

\begin{proof}
We first show that the payoff $P(\A) = w(p)$. Consider the relation between $\Pr(E_x^j|V_x^k)$ and $\Pr(E_x^j|W_x^k)$. If time $x$ is the $\ell$-th best item overall, then it must be a $k$-potential for some $k\leq \ell$. Moreover, we have $\Pr(V_x^k | W_x^\ell) = \binom{\ell-1}{k-1} x^{k-1} (1-x)^{\ell-k}$ (by convention $0^0=1$). Then
\begin{align*}
\Pr(E_x^j|W_x^\ell) & = \txts\sum_{k=1}^{\ell} \Pr(E_x^j|V_x^k \wedge W_x^\ell) \Pr(V_x^k | W_x^\ell) = \sum_{k=1}^{\ell} \Pr(E_x^j|V_x^k) \Pr(V_x^k | W_x^\ell) \\
& = \txts\sum_{k=1}^{\ell} \binom{\ell-1}{k-1} x^{k-1} (1-x)^{\ell-k} p_{j|k}(x).
\end{align*}

Let $\mathbb{I}_\ell$ be the indicator that the $\ell$-th best item is selected. Since the probability density function of each arrival time is uniform in $[0,1]$, the payoff of the algorithm is
\begin{align*}
P(\A) & = \txts\expct{\sum_{\ell=1}^{K} \mathbb{I}_\ell} = \txts\sum_{\ell=1}^{K} \expct{\mathbb{I}_\ell} = \sum_{\ell=1}^{K} \Pr(\ell\text{-th best item selected}) \\
& = \txts\sum_{\ell=1}^{K} \sum_{j=1}^{J} \int_{0}^{1} 1 \cdot \Pr(E_x^j | W_x^\ell) d x \\
& = \txts\sum_{j=1}^{J} \int_{0}^{1} \sum_{\ell=1}^{K} \sum_{k=1}^{\ell} \binom{\ell-1}{k-1} x^{k-1} (1-x)^{\ell-k} p_{j|k}(x) d x \\
& = \txts\sum_{j=1}^{J} \sum_{k=1}^{K} \int_{0}^{1} \left(\sum_{\ell=k}^{K} \binom{\ell-1}{k-1} (1-x)^{\ell-k} \right) x^{k-1} p_{j|k} (x) dx.
\end{align*}

For the constraints, by Lemma~\ref{lemma:k-trick} we have 
$p_{J|k} (x) = \Pr(E_x^J|V_x^k) \leq \Pr(\overline{Z_x^J}) = 1 - \int_{0}^{x} \frac{1}{y} \sum_{\ell=1}^{K} p_{J|\ell}(y) d y$,
and 
$p_{j|k} (x) = \Pr(E_x^j|V_x^k) \leq \Pr(Z_x^{j+1}\wedge\overline{Z_x^j}) = \Pr(Z_x^{j+1}) - \Pr(Z_x^j) 
= \int_{0}^{x} \frac{1}{y} \sum_{\ell=1}^{K} [p_{(j+1)|\ell}(y) - p_{j|\ell} (y)] d y$
for $1\leq j < J$,
where the second last equality follows since $Z_x^j$ implies $Z_x^{j+1}$.
\end{proof}


\begin{proposition} [Relation between $\LP_n(J,K)$ and $\CP(J,K)$] \label{prop:p_cp}
Let $\mathcal{P}_n^*$ and $\CP^*$ be the optimal values of $\LP_n(J,K)$ and $\CP(J,K)$, respectively. Then, for every $\epsilon>0$, there exists $N\in\N$ such that $\CP^* \geq \mathcal{P}_n^*-\epsilon$ for all $n \geq N$.
\end{proposition}

\begin{proof}
Our proof strategy is as follows. We start from an optimal solution $y$ of $\LP_n(J,K)$, with objective value $v_n(y) = \mathcal{P}_n^*$. Our goal is to construct a feasible solution $p$ of $\CP(J,K)$ such that the difference $v_n(y) - w(p)$ of objective values is small for sufficiently large $n$. The idea is to transform $y$ into $p$ by interpolation. However, a piecewise continuous solution directly constructed from $y$ might not be feasible for $\CP(J,K)$; intuitively, the constructed functions at points $x$ for small (constant) $i\in[n]$ may violate the constraints by large (constant) values. We introduce two intermediate solutions: $z$ with respect to $\LP_n(J,K)$ and $r$ with respect to $\CP(J,K)$. To avoid constraint violation due to small $i$, the solution $z$ is obtained by shifting $y$ by a distance of $s\leq n$ such that $z_{j|k}(i) = 0$ for all $i< s$. Then, we construct $r$ from $z$ by interpolation, which is not necessarily feasible to $\CP(J,K)$ but can only violate the constraints for $i\geq s$. Finally, we reduce $r$ by some multiplicative factors and obtain a feasible solution $p$ of $\CP(J,K)$. The parameter $s$ is carefully selected such that the difference $v_n(y) - w(p)$ remains small.

Let $y$ be an optimal solution of $\LP_n(J,K)$. Let $s$ with $3K \leq s \leq n$ be an integer to be determined later. Define a solution $z$ to $\LP_n(J,K)$ as follows: for each $1\leq j \leq J$ and $1\leq k \leq K$, set $z_{j|k}(i) := 0$ for $1\leq i < s$ and $z_{j|k}(i) := y_{j|k}(i-s+1)$ for $s \leq i \leq n$. For $1\leq i < s$, obviously the constraints hold for $z_{j|k}(i)$. Suppose $s \leq i \leq n$, then we have
\begin{align*}
z_{J|k}(i) & = y_{J|k}(i-s+1) \leq 1 - \txts\sum_{m=1}^{i-s} \frac{1}{m}\sum_{\ell=1}^{K} y_{J|\ell}(m) \\
& = 1 - \txts\sum_{m=s}^{i-1} \frac{1}{m}\sum_{\ell=1}^{K} z_{J|\ell}(m) = 1 - \txts\sum_{m=1}^{i-1} \frac{1}{m}\sum_{\ell=1}^{K} z_{J|\ell}(m)
\end{align*}
and
\begin{align*}
z_{j|k}(i) & = y_{j|k}(i-s+1) \leq \txts\sum_{m=1}^{i-s} \frac{1}{m} \sum_{\ell=1}^{K} [y_{(j+1)|\ell}(m) - y_{j|\ell}(m)] \\
& = \txts\sum_{m=1}^{i-1} \frac{1}{m} \sum_{\ell=1}^{K} [z_{(j+1)|\ell}(m) - z_{j|\ell}(m)]
\end{align*}
for $1 \leq j < J$. Therefore $z$ is a feasible solution of $\LP_n(J,K)$. Next we analyze $v_n(z)$. First observe that
\begin{align*}
v_n(z) & = \txts\sum_{j=1}^J \sum_{k=1}^{K} \sum_{i=s}^n \frac{1}{n} \sum_{\ell=k}^{K} \frac{\binom{n-i}{\ell-k} \binom{i-1}{k-1}}{\binom{n-1}{\ell-1}} z_{j|k}(i) \\
& = \txts\sum_{j=1}^J \sum_{k=1}^{K} \sum_{i=1}^{n-s+1} \frac{1}{n} \sum_{\ell=k}^{K} \frac{\binom{n-i-s+1}{\ell-k} \binom{i+s-2}{k-1}}{\binom{n-1}{\ell-1}} y_{j|k}(i).
\end{align*}
For $1\leq k \leq \ell\leq K$ and $1\leq i < n-K$ and positive integer $m$ with $K \leq m \leq n-i$, we have
\begin{align*}
\txts \frac{\binom{n-i-m}{\ell-k} \binom{i+m-1}{k-1}}{\binom{n-i}{\ell-k} \binom{i-1}{k-1}}
& \geq \txts \frac{\binom{n-i-m}{\ell-k}}{\binom{n-i}{\ell-k}}
\geq \txts \frac{(n - i - m - \ell+k)^{\ell-k}}{(n-i)^{\ell-k}}
= \txts \left(1 - \frac{m+\ell-k}{n-i} \right)^{\ell-k} \\
& \geq 1 - \txts\frac{(\ell-k)(m+\ell-k)}{n-i} \geq 1 - \frac{K(m+K)}{n-i} \geq 1 - \frac{2Km}{n-i}.
\end{align*}
Then it follows that
\begin{align}
\txts \frac{\binom{n-i}{\ell-k} \binom{i-1}{k-1}}{\binom{n-1}{\ell-1}} - \frac{\binom{n-i-m}{\ell-k} \binom{i+m-1}{k-1}}{\binom{n-1}{\ell-1}} \leq \frac{2Km}{n-i} \cdot \frac{\binom{n-i}{\ell-k} \binom{i-1}{k-1}}{\binom{n-1}{\ell-1}}. \label{eq:lp_equiv}
\end{align}
Since $0\leq y_{j|k}(i) \leq 1$ for all $i, j, k$ and $\sum_{k=1}^{K}\sum_{\ell=k}^{K}\frac{\binom{n-i}{\ell-k} \binom{i-1}{k-1}}{\binom{n-1}{\ell-1}} = K$, we have
\begin{align*}
v_n(y) - v_n(z) & \leq \txts\sum_{j=1}^J \sum_{k=1}^{K} \sum_{i=1}^{n-s^2} \frac{1}{n} \sum_{\ell=k}^{K} \left( \frac{\binom{n-i}{\ell-k} \binom{i-1}{k-1}}{\binom{n-1}{\ell-1}} - \frac{\binom{n-i-s+1}{\ell-k} \binom{i+s-2}{k-1}}{\binom{n-1}{\ell-1}} \right) y_{j|k}(i) \\
& \qquad + \txts\sum_{j=1}^J \sum_{k=1}^{K} \sum_{i=n - s^2 + 1}^n \frac{1}{n} \sum_{\ell=k}^{K} \frac{\binom{n-i}{\ell-k} \binom{i-1}{k-1}}{\binom{n-1}{\ell-1}} y_{j|k}(i) \\
& \leq \txts\sum_{j=1}^J \sum_{k=1}^{K} \sum_{i=1}^{n-s^2} \frac{1}{n} \sum_{\ell=k}^{K} \frac{2Ks}{n-i} \cdot \frac{\binom{n-i}{\ell-k} \binom{i-1}{k-1}}{\binom{n-1}{\ell-1}} y_{j|k}(i) + \frac{JK s^2}{n} \\
& \leq \txts \frac{2K}{s} v_n(y) + \frac{JK s^2}{n} \leq \frac{2JK}{s} + \frac{JK s^2}{n},
\end{align*}
where the second inequality follows from~\eqref{eq:lp_equiv}, the third from $i \leq n-s^2$ and the last from $v_n(y) = \calP_n^* \leq J$. Then we have
\begin{align*}
v_n(z) \geq \calP_n^* - \txts \frac{JK s^2}{n} - \frac{2JK}{s}.
\end{align*}


For each $j$ and $k$, define function $r_{j|k}(x)$ as follows: set $r_{j|k}(x) := 0$ when $0 \leq x \leq \frac{s}{n}$ and $r_{j|k}(x) := z_{j|k}(i)$ when $\frac{i}{n} < x \leq \frac{i+1}{n}$ for $s \leq i \leq n-1$. Then we have
\begin{align*}
w(r) & = \txts\sum_{j=1}^{J} \sum_{k=1}^{K} \int_{0}^{1} \left(\sum_{\ell=k}^{K} \binom{\ell-1}{k-1} (1-x)^{\ell-k} \right) x^{k-1} r_{j|k} (x) dx \\
& = \txts\sum_{j=1}^{J} \sum_{k=1}^{K} \sum_{m=s}^{n-1} \int_{\frac{m}{n}}^{\frac{m+1}{n}} \left(\sum_{\ell=k}^{K} \binom{\ell-1}{k-1} (1-x)^{\ell-k} \right) x^{k-1} z_{j|k} (m) dx \\
& \geq \txts\sum_{j=1}^{J} \sum_{k=1}^{K} \sum_{m=s}^{n-1} \frac{1}{n} \left(\sum_{\ell=k}^{K} \binom{\ell-1}{k-1} (1-\frac{m+1}{n})^{\ell-k} \right) (\frac{m}{n})^{k-1} z_{j|k} (m).
\end{align*}
We wish to show $\frac{\binom{n-m}{\ell-k} \binom{m-1}{k-1}}{\binom{n-1}{\ell-1}} = \binom{\ell-1}{k-1} (1-\frac{m+1}{n})^{\ell-k} (\frac{m}{n})^{k-1} + O(\frac{1}{n})$ for each $k, \ell, m$.
Observe that $\frac{\ell-1}{n-\ell+1} \leq \frac{K}{n-K} \leq \frac{2K}{n}$ whenever $n\geq 2K$. We have
\begin{align*}
\txts \frac{\binom{n-m}{\ell-k} \binom{m-1}{k-1}}{\binom{n-1}{\ell-1}}
& = \txts\binom{\ell-1}{k-1} \cdot \frac{(n-m)!}{(n-m-\ell+k)!} \cdot \frac{(m-1)!}{(m-k)!} \cdot \frac{(n-\ell)!}{(n-1)!} \\
& \leq \txts \binom{\ell-1}{k-1} (n-m)^{\ell-k} m^{k-1} (\frac{1}{n-\ell+1})^{\ell-1} \\
& = \txts \binom{\ell-1}{k-1} m^{k-1} \left\{\sum_{u=0}^{\ell-k} \binom{\ell-k}{u} (n-m-1)^u \sum_{u=0}^{\ell-1} \binom{\ell-1}{u} (\frac{1}{n})^u (\frac{\ell-1}{n(n-\ell+1)})^{\ell-u-1} \right\} \\
& \leq \txts \binom{\ell-1}{k-1} m^{k-1} \left\{\left[(n-m+1)^{\ell-k} + \txts K \cdot 2^K (n-m+1)^{\ell-k-1} \right] \right. \\
& \qquad \cdot \txts \left. \left[ (\frac{1}{n})^{\ell-1} + K \cdot 2^K (\frac{2K}{n})^{\ell} \right] \right\} \\
& \leq \txts \binom{\ell-1}{k-1} m^{k-1} \txts \left\{ (n-m+1)^{\ell-k} (\frac{1}{n})^{\ell-1} + \frac{c_0}{n^k \binom{\ell-1}{k-1}} \right\} \\
& \leq \txts \binom{\ell-1}{k-1} (1-\frac{m+1}{n})^{\ell-k} (\frac{m}{n})^{k-1} + \frac{c_0}{n},
\end{align*}
where $c_0 = c_0(K)$ is some constant. Then we have
\begin{align*}
w(r) & \geq \txts \sum_{j=1}^J \sum_{k=1}^{K} \sum_{m=s}^{n-1} \frac{1}{n} \sum_{\ell=k}^{K} \left(\frac{\binom{n-m}{\ell-k} \binom{m-1}{k-1}}{\binom{n-1}{\ell-1}}-\frac{c_0}{n} \right) z_{j|k}(m) \\
& = v_n(z) - \txts \sum_{j=1}^J \sum_{k=1}^{K} \frac{1}{n} \sum_{\ell=k}^{K} \frac{\binom{0}{\ell-k} \binom{n-1}{k-1}}{\binom{n-1}{\ell-1}} z_{j|k}(n) \\
& \qquad - \txts \sum_{j=1}^J \sum_{k=1}^{K} \sum_{m=s}^{n-1} \frac{1}{n} \sum_{\ell=k}^{K} \frac{c_0}{n} z_{j|k}(m) \\
& \geq v_n(z) - \txts \frac{J K}{n} - \frac{J K^2 c_0}{n} \geq \mathcal{P}_n^* - \frac{JK(s^2+c)}{n} - \frac{2JK}{s},
\end{align*}
where $c := K c_0 + 1$ is a constant.

Suppose $\frac{s}{n} < x \leq 1$. Let $i$ be the integer such that $\frac{i}{n} < x \leq \frac{i+1}{n}$. Observe that for each $j$ and $k$, we have
\begin{align*}
\txts\int_0^x \frac{r_{j|k}(y)}{y} d y & = \txts \sum_{m=s}^{i-1} \int_{\frac{m}{n}}^{\frac{m+1}{n}} \frac{z_{j|k}(m)}{y} d y + \int_{\frac{i}{n}}^x \frac{z_{j|k}(i)}{y} d y \\
& \leq \txts \sum_{m=s}^{i-1}\frac{z_{j|k}(m)}{m/n} \cdot \frac{1}{n} + \frac{z_{j|k}(i)}{i/n} \cdot \frac{1}{n} = \sum_{m=s}^i \frac{z_{j|k}(m)}{m},
\end{align*}
and
\begin{align*}
\txts\int_0^x \frac{r_{j|k}(y)}{y} d y & = \txts\sum_{m=s}^{i-1} \int_{\frac{m}{n}}^{\frac{m+1}{n}} \frac{z_{j|k}(m)}{y} d y + \int_{\frac{i}{n}}^x \frac{z_{j|k}(i)}{y} d y \\
& \geq \txts \sum_{m=s}^{i-1}\frac{z_{j|k}(m)}{(m+1)/n} \cdot \frac{1}{n} \geq \frac{s}{s+1} \sum_{m=s}^{i-1}\frac{z_{j|k}(m)}{m},
\end{align*}
where the last inequality follows from $\frac{m}{m+1} \geq \frac{s}{s+1}$ for $m\geq s$.
Then we have
\begin{align}
r_{J|k}(x) = z_{J|k}(i) & \leq 1 - \txts\sum_{m=s}^{i-1} \frac{1}{m}\sum_{\ell=1}^{K} z_{J|\ell}(m) \nonumber \\
& \leq 1 - \txts\int_{0}^{x} \frac{1}{y} \sum_{\ell=1}^{K} r_{J|\ell}(y) d y + \frac{1}{i}\sum_{\ell=1}^{K} z_{J|\ell}(i) \nonumber \\
& = 1 - \txts\int_{0}^{x} \frac{1}{y} \sum_{\ell=1}^{K} r_{J|\ell}(y) d y + \frac{1}{i}\sum_{\ell=1}^{K} r_{J|k}(x) \nonumber \\
& \leq \txts\frac{s+K}{s} - \int_{0}^{x} \frac{1}{y} \sum_{\ell=1}^{K} r_{J|\ell}(y) d y. \label{eq:equiv_rJ}
\end{align}
For $1\leq j < J$, we have
\begin{align*}
r_{j|k}(x) & = z_{j|k}(i) \leq \txts\sum_{m=s}^{i-1} \frac{1}{m} \sum_{\ell=1}^{K} [z_{(j+1)|\ell}(m) - z_{j|\ell}(m)] \\
& \leq \txts\frac{s+1}{s} \int_{0}^{x}\frac{1}{y} \sum_{\ell=1}^{K} r_{(j+1)|\ell}(y) d y - \txts\int_{0}^{x} \frac{1}{y} \sum_{\ell=1}^{K} r_{j|\ell}(y) d y + \frac{1}{s}\sum_{\ell=1}^{K} r_{j|\ell}(x).
\end{align*}
Summing up the above inequalities over $k$ yields
\begin{align*}
\txts\sum_{\ell=1}^{K} r_{j|\ell}(x) \leq \frac{K(s+1)}{s-K} \int_{0}^{x}\frac{1}{y} \sum_{\ell=1}^{K} r_{(j+1)|\ell}(y) d y,
\end{align*}
then it follows that
\begin{align}
r_{j|k}(x) & \leq \txts \int_{0}^{x}\frac{1}{y} \sum_{\ell=1}^{K} [r_{(j+1)|\ell}(y) - r_{j|\ell}(y) ] d y + \txts\frac{K+1}{s-K} \int_{0}^{x}\frac{\sum_{\ell=1}^{K} r_{(j+1)|\ell}(y)}{y}  d y. \label{eq:equiv_rj}
\end{align}

Now we define a solution $p$ for $\CP(J,K)$ as follows.
For each $j$ and $k$, set $p_{j|k}(x) := (1-(J-j+1)\delta) \cdot r_{j|k}(x)$ for $x\in[0,1]$, where $\delta\in(0,\frac{1}{J})$ is determined later. Note that $p_{j|k}(x) = 0$ for $0\leq x \leq \frac{s}{n}$. Since $r_{j|k}(x)\leq 1$ for all $j$ and $k$, and $\sum_{k=1}^{K} \sum_{\ell=k}^{K} \binom{\ell-1}{k-1} (1-x)^{\ell-k} x^{k-1} = K$, we have
\begin{align*}
w(p) & = w(r) - \txts \sum_{j=1}^{J} \sum_{k=1}^{K} \int_{\frac{s}{n}}^{1} \left(\sum_{\ell=k}^{K} \binom{\ell-1}{k-1} (1-x)^{\ell-k} \right) x^{k-1} (J-j+1) \delta r_{j|k} (x) dx \\
& \geq w(r) - J^2 K \delta \geq \mathcal{P}_n^* - \txts \frac{JK(s^2+c)}{n} - \frac{2JK}{s} - J^2 K \delta.
\end{align*}

For $0\leq x \leq \frac{s}{n}$, obviously the constraints of $\CP(J,K)$ hold for $p_{j|k}(x)$. Suppose $\frac{s}{n} < x \leq 1$. Then from~\eqref{eq:equiv_rJ} we have
\begin{align*}
p_{J|k}(x) & = (1-\delta) r_{J|k}(x)
\leq \txts\frac{(1-\delta)(s+K)}{s} - \int_{0}^{x} \frac{1}{y} \sum_{\ell=1}^{K} p_{J|\ell}(y) d y \\
& \leq \txts 1 - \int_{0}^{x} \frac{1}{y} \sum_{\ell=1}^{K} p_{J|\ell}(y) d y,
\end{align*}
assuming $\delta \geq \frac{K}{s}$ (and hence $\frac{(1-\delta)(s+K)}{s}\leq 1$). For $1\leq j < J$, from~\eqref{eq:equiv_rj} we have
\begin{align*}
p_{j|k}(x) & = (1-(J-j+1)\delta) r_{j|k}(x) \\
& \leq \txts \int_{0}^{x}\frac{1}{y} [\sum_{\ell=1}^{K} (1-(J-j)\delta - \delta) r_{(j+1)|\ell}(y) - (1-(J-j+1)\delta) r_{j|\ell}(y)] d y \\
& \qquad + \txts\frac{K+1}{s-K} \int_{0}^{x}\frac{1}{y} \sum_{\ell=1}^{K} r_{(j+1)|\ell}(y) d y \\
& \leq \txts \int_{0}^{x}\frac{1}{y} [\sum_{\ell=1}^{K} p_{(j+1)|\ell}(y) - p_{j|\ell}(y)] d y
+ \txts\left(\frac{K+1}{s-K}-\delta\right) \int_{0}^{x}\frac{1}{y} \sum_{\ell=1}^{K} r_{(j+1)|\ell}(y) d y \\
& \leq \txts \int_{0}^{x}\frac{1}{y} [\sum_{\ell=1}^{K} p_{(j+1)|\ell}(y) - p_{j|\ell}(y)] d y,
\end{align*}
assuming $\delta \geq \frac{3K}{s}$ (and hence $\frac{K+1}{s-K}-\delta \leq 0$). Therefore $p$ is a feasible solution for $(CP)$, whenever $\frac{3K}{s} \leq \delta < \frac{1}{J}$, with objective value $w(p) \geq \mathcal{P}_n^* - \frac{JK(s^2+c)}{n} - \frac{2JK}{s} - J^2 K \delta$. Set $s := \lceil \sqrt[3]{n} \, \rceil$ and $\delta := \frac{3K}{s} \leq \frac{3K}{\sqrt[3]{n}}$. For every $\epsilon \in(0,1)$, let $N := \lceil \max \{ \frac{512 J^6 K^6}{\epsilon^3}, c^{3/2}\} \rceil$. Then for all $n\geq N$, we have $\delta \leq \frac{3 \epsilon}{8 J^2 K} < \frac{1}{J}$ and
\begin{align*}
\mathcal{P}_n^* - \CP^* & \leq \mathcal{P}_n^* - w(p) \leq \txts \frac{JK(s^2+c)}{n} + \frac{2JK}{s} + J^2 K \delta \\
& \leq \txts \frac{3JK n^{2/3}}{n} + \frac{2JK}{\sqrt[3]{n}} + \frac{3 J^2 K^2}{\sqrt[3]{n}} \leq \frac{8 J^2 K^2}{\sqrt[3]{N}} \leq \epsilon,
\end{align*}
as required.
\end{proof}

\section{A Primal-Dual Method for Finding Thresholds} \label{sec:optalg}

We give a primal-dual procedure that finds appropriate thresholds
for which the $(J,K)$-Threshold Algorithm corresponds
to an optimal solution in the continuous linear program $\CP(J,K)$.
To illustrate our primal-dual method, we first
consider the special case $K=1$ as described in Theorem~\ref{th:main1}; the general case is given in Section~\ref{sec:k-best}.

The $J$-Threshold Algorithm is a special case
with $t_j := \tau_{j,1}$, and recall that
any algorithm in the class $\scrA$ corresponds
to a feasible solution in the following primal continuous LP:
\ignore{
We introduce a threshold algorithm $\D$ for the general $J$-choice secretary problem. The algorithm takes a vector $T_J=(t_1,\ldots,t_J)$ with $0< t_J \leq \cdots \leq t_1 < 1$ as input. In time interval $[0,t_J)$, the algorithm cannot select an item. At each time $t_j$ for $1\leq j \leq J$, a new quota is allocated to $\D$; that is, one more choice is allowed starting from $t_j$. By Proposition~\ref{prop:finite_can} there are finite potentials in $[t_J,1]$. The algorithm $\D$ selects a potential (starting from $t_J$) whenever there are available quotas. Note that $\D$ can select at most $J-j$ items before time $t_j$ for $1\leq j \leq J$ and at most $J$ items by time 1.

\noindent\begin{tabular}{|p{\textwidth}|}
\hline
\textbf{The Threshold Algorithm $\D=\D(T_J)$} \\
\textbf{Input:} A vector of times $T_J=(t_1,\ldots,t_J)$ with $0 \leq t_J \leq \cdots \leq t_1 \leq 1$. \\
Let $t_{0} = 1$. \\
In time interval $[0,t_J]$, observe arriving items (but select none). \\
For $i=1,2,\ldots,J$ \\
\qquad In time interval $[t_{J-i+1}, t_{J-i}]$, upon arrival of an item at time $x$, select $x$ if \\
\qquad\qquad (1) $x$ is a potential; and \\
\qquad\qquad (2) at most $i-1$ items have been selected. \\
\hline
\end{tabular}
}
\begin{align*}
\CP(J) \qquad \max \qquad & w(p) = \txts\sum_{j=1}^J \int_0^1 p_j(x) dx \\
\text{s.t.} \qquad & p_j(x) \leq \txts\int_0^x \frac{1}{y} [p_{j+1}(y) - p_j(y)] dy, \qquad \forall x\in[0,1], 1\leq j < J \\
 & p_J(x) \leq 1 - \txts\int_0^x \frac{p_J(y)}{y} d y, \qquad \forall x\in[0,1] \\
& p_j(x) \geq 0, \qquad \forall x\in[0,1], j\in [J].
\end{align*}


The dual LP for $\CP(J)$ is as follows (see~\cite{Levinson66} for details
on primal-dual continuous LP):
\begin{align*}
\CD(J) \qquad \min \qquad & \txts\int_0^1 q_J(x) d x \\
\text{s.t.} \qquad & q_1(x) + \txts\frac{1}{x} \int_x^1 q_1(y) d y \geq 1, \qquad \forall x\in[0,1] \\
 & q_j(x) + \txts\frac{1}{x} \int_x^1 [q_j(y) - q_{j-1}(y)] d y \geq 1, \qquad \forall x\in[0,1], 1 < j \leq J \\
& q_j(x) \geq 0, \qquad \forall x\in[0,1], j\in[J].
\end{align*}

\noindent \textbf{Weak Duality.}  Similar to normal LP, for any feasible
primal $p$ and dual $q$, the value of the primal objective is at most
that of the dual objective.  Moreover, if their objective values
are equal, then both are optimal.  We also have the
following complementary slackness conditions.

\begin{fact} [Complementary Slackness Conditions] \label{fact:css_clp}
Let $p=(p_1,\ldots,p_J)$ and $q=(q_1,\ldots,q_J)$ be feasible solutions of $\CP(J)$ and $\CD(J)$, respectively. Then, $p$ and $q$ are primal and dual optimal, respectively, if they satisfy the following conditions $\forall x \in [0,1]$:
\begin{align*}
& \txts\left( p_J(x) + \int_0^x \frac{1}{y} p_J(y) dy - 1\right) q_J(x) = 0 \\
& \txts\left( p_j(x) + \int_0^x \frac{1}{y} [p_j(y) - p_{j+1}(y)] d y \right) q_j(x) = 0, \qquad 1\leq j < J \\
& \txts\left( q_1(x) + \frac{1}{x} \int_x^1 q_1(y) d y - 1 \right) p_1(x) = 0 \\
& \txts\left( q_j(x) + \frac{1}{x} \int_x^1 [q_j(y) - q_{j-1}(y)] d y - 1 \right) p_j(x) = 0, \qquad 1 < j \leq J.
\end{align*}
\end{fact}

\noindent \textbf{Primal-Dual Method.}  We start from a primal feasible solution $p$ corresponding
to a $J$-Threshold Algorithm, whose thresholds are to be determined.
We can determine the values of the thresholds one by one in order
to construct a dual $q$ such that complementary slackness conditions
hold, which implies that with those found thresholds the $J$-Threshold
Algorithm is optimal.

\noindent \textbf{(1) Forming Feasible Primal Solution $p$.}
Suppose $p$ is the (feasible) primal corresponding to the
$J$-Threshold Algorithm with thresholds 
$0<t_J \leq t_{J-1} \leq \cdots \leq t_1 \leq 1$.
We denote $E_x^j$ to be the event that item at $x$ is selected
by using quota $Q_j$ (where quotas with larger $j$'s are used first),
$V_x$ to be the event that $x$ is a potential, 
and $Z_x^j$ to be the event that at time $x$, quota $Q_j$ has
already been used (and so have the quotas with indices larger than $j$).
For notational convenience, $Z_x^{J+1}$ is the whole sample space, i.e.,
an always true event.

For each $j \in [J]$, consider the conditional probability
$\Pr(E_x^j |Z_x^{j+1}\wedge\overline{Z_x^j}\wedge V_x)$
of the event that item at $x$ is selected
by using quota $Q_j$, given that $x$ is a potential and quota $Q_j$ is the next available quota at time $x$.  By definition of the Threshold
Algorithm, this conditional probability is 0 if $x < t_j$ and is 1 if $x \geq t_j$.  Hence, we have the following.
\begin{align*}
\txts \Pr(E_x^j|Z_x^{j+1}\wedge\overline{Z_x^j}\wedge V_x)  = \txts\frac{\Pr(E_x^j|V_x)}{\Pr(Z_x^{j+1}\wedge\overline{Z_x^j} | V_x)} = \txts \frac{p_j(x)}{\Pr(Z_x^{j+1}\wedge\overline{Z_x^j} | V_x)} 
 = 
\begin{cases}
 0, & 0 \leq x < t_j \\
 1, & t_j \leq x \leq 1,
\end{cases}
\end{align*}
where from independence of $V_x$ and $Z_x^j$, and Lemma~\ref{lemma:k-trick}, we have:
\begin{align*}
\Pr(Z_x^{j+1}\wedge\overline{Z_x^j} | V_x) = \Pr(Z_x^{j+1}\wedge\overline{Z_x^j}) 
=\begin{cases}
 \int_0^x \frac{1}{y} [p_{j+1}(y) - p_j(y)] d y, & 1 \leq j < J \\
 1-\int_0^x \frac{1}{y} p_J(y) d y, & j = J.
\end{cases}
\end{align*}

This implies that in the primal $\CP(J)$,
the $j$-th constraint is equality
in the range $[t_j, 1]$, but might be strict inequality
in the range $[0, t_j)$ (hence forcing $q_j$ to 0);
the function $p_j$ is zero in the range $[0,t_j)$, but
might be strictly positive in the range $[t_j, 1]$ (hence
forcing equality for the $j$-th constraint in dual).

\noindent \textbf{(2) Finding Feasible Dual $q$ to Satisfy
Complementary Slackness.}  To ensure that a dual solution $q$
satisfies complementary slackness together with the above primal $p$,
we require the following for each $j \in [J]$, where for
notational convenience we write $q_0 \equiv 0$.

\begin{equation} \label{eq:slack}
\begin{cases}
q_j(x) = 0, &\qquad x \in [0, t_j] ;\\
q_j(x) + \frac{1}{x} \int_x^1 [q_j(y) - q_{j-1}(y)] d y = 1, &\qquad
x \in  [t_j, 1].
\end{cases}
\end{equation}

The astute reader might notice that
we have imposed an extra condition $q_j(t_j)=0$.  This
will ensure that as long as (\ref{eq:slack}) is satisfied by some
non-negative $q_j$, the $j$-th constraint in $\CD(J)$ is also automatically
satisfied.  For $x \in [t_j,1]$,
the constraint is clearly satisfied with equality;
for $x \in [0, t_j)$, observing that
both $q_j$ and $q_{j-1}$ vanishes below $t_j$
the left hand side reduces to $\frac{1}{x} \int_{t_j}^1 [q_j(y) - q_{j-1}(y)] dy$,
which is larger than $q_j(t_j) + \frac{1}{t_j} \int_{t_j}^1 [q_j(y) - q_{j-1}(y)] = 1$.

As we shall see soon, in the recursive equations (\ref{eq:slack}),
the function $q_1$ and the threshold $t_1$ does not depend on $J$.
In particular, the thresholds $t_j$'s and functions $q_j$'s found for $\CD(J)$
can be used to extend to the solution for $\CD(J+1)$.  This explains
the nice structure of the solution that appears in Theorem~\ref{th:main1}.

\noindent \textbf{Objective Value.} The objective
value of $\CD(J)$ is $\int_0^1 q_J(y) dy = \int_{t_J}^1 q_J(y) dy$.
From the second equation of (\ref{eq:slack}) evaluating at $x=t_j$, we have
the recursive definition $\int_{t_j}^1 q_j(y) dy = \int_{t_{j-1}}^1 q_{j-1}(y) dy + t_j$, which immediately implies
that the objective value for $\CD(J)$ is $\sum_{j=1}^J t_j$, as stated
in Theorem~\ref{th:main1}.  Hence, it suffices to show the 
existence of dual functions as required in (\ref{eq:slack}).

\begin{lemma} [Existence of Feasible Dual Satisfying Complementary Slackness] \label{lemma:slack}
There is a procedure to generate an increasing sequence $\{\theta_j\}_{j \geq 1}$
of rational numbers producing $t_j := \frac{1}{e^{\theta_j}}$,
and a sequence $\{q_j:[0,1] \rightarrow \R^+\}_{j \geq 1}$ of non-negative functions that satisfy (\ref{eq:slack}).
\end{lemma}

\begin{proof}
We show the existence result by induction; our induction proof
actually gives a method to generate such $t_j$'s and $q_j$'s.
We explicitly describe the method in Section~\ref{sec:generate},
and it can be seen that the time to generate the first $J$ thresholds
is $O(J^3)$.

For convenience, we denote
$q_0(x) \equiv 0$ and set $\theta_0 := 0$ and $t_0 := 1$.
Suppose for some $j \geq 1$ we have constructed the function
$q_{j-1}$ which is continuous and can be positive only in $[t_{j-1}, 1]$.  We
next wish to find continuous function $q_j$ and threshold $t_j < t_{j-1}$
satisfying (\ref{eq:slack}).  If such $q_j$ and $t_j$ exist,
then we must have the following for $x \in [t_j, 1]$:
\begin{align*}
\txts q_j(x) + \frac{1}{x} \int_x^1 [q_j(y) - q_{j-1}(y)]dy = 1 \\
\txts x q_j(x) + \int_x^1 [q_j(y) - q_{j-1}(y)]dy = x \\
\txts (x q_j'(x) + q_j(x)) - q_j(x)  + q_{j-1}(x) = 1 \\
\txts q_j'(x) = \frac{1}{x}(1 - q_{j-1}(x)). 
\end{align*}
Since $q_j(1) = 1$, we must have
\begin{equation} \label{eq:q_j}
\txts q_j(x) = 1 + \ln x + \int_x^1 \frac{q_{j-1}(y)}{y} dy, \quad \forall x \in [t_j, 1]
\end{equation}
 
To show that both $q_j$ and $t_j$ exist, we need a stronger induction hypothesis.
Hence, we first explicitly solve for $q_1$ and $t_1$, and state what properties
we can assume.  Since $q_0(x) \equiv 0$, we have $q_1(x) = 1 + \ln x$ on $[t_1,1]$.
In order to have $q_1(t_1) = 0$, we must have $t_1 := \frac{1}{e}$ and $\theta_1 := 1$.
We give our induction hypothesis, which is true for $j=1$.

\noindent \textbf{Induction Hypothesis.} Suppose for some $j \geq 1$,
there exist functions $\{q_i\}_{i=0}^j$ and thresholds $\{t_i\}_{i=0}^j$
satisfying (\ref{eq:slack}) such that the following holds.

\begin{compactitem}
\item[1.]  The function $q_j$ is non-negative and continuous.
\item[2.] There exists an increasing sequence $\{\theta_i\}_{i=0}^j$
of rational numbers that defines the thresholds $t_i := \exp(- \theta_i)$
such that $q_j$ is 0 on $[0, t_j]$ and between successive thresholds, $q_j(x)$
is given by a polynomial in $\ln x$ with rational coefficients.
\item[3.] For $x \in (t_j, 1)$, $q_j(x) > q_{j-1}(x)$.
\end{compactitem}

We next show the existence of $q_{j+1}$ and $t_{j+1}$.  

\noindent \textbf{Finding $q_{j+1}$.} From (\ref{eq:q_j}),
$q_{j+1}(x)$ must agree on $[t_{j+1},1]$ with the function $q(x)$ given by
$q(x) = 1 + \ln x + \int_x^1 \frac{q_j(y)}{y} dy$, which is continuous.

We first check that we can set $q_{j+1}(x) := q(x)$ for $x \in [t_j, 1]$.
Since from the induction hypothesis we have $q_{j} > q_{j-1}$ on $(t_j, 1)$,
we immediately have $\forall x \in [t_j, 1)$,
$q(x) = 1 + \ln x + \int^1_{x} \frac{q_{j}(y)}{y} dy > 1 + \ln x + \int^1_{x} \frac{q_{j-1}(y)}{y} dy = q_j(x)$.
In particular, we have $q(t_j) > q_j(t_j) = 0$,
and also $q(x) \geq q_j(x) \geq 0$ for $x \in [t_j,1]$.
%

From the induction hypothesis on $q_j$, we can conclude
that between successive thresholds in $[t_j,1]$, $q_{j+1}(x)$
can also be represented by a polynomial in $\ln x$ with rational coefficients.
Hence, it follows that $d_j := \int^1_{t_j} \frac{q_j(y)}{y} dy$ is rational.

\noindent \textbf{Finding $t_{j+1}$.} We next consider the behavior of $q$ for $x \leq t_j$.
Observe that in this range, $q(x) = 1 + \ln x + d_j$, which is a polynomial
in $\ln x$ with rational coefficients, and strictly increasing in $x$.
Moreover, we have $q(t_j) > 0$ and as $x$ tends to 0, $q(x)$ tends to negative infinity.
Hence, there is a unique $t_{j+1} \in (0, t_j)$ such that $q(t_{j+1}) = 0$;
we set $t_{j+1} := \exp(-\theta_{j+1})$, where $\theta_{j+1} := 1 + d_j$,
which is rational.

Hence, we can set $q_{j+1}(x) := q(x)$ for $x \in [t_{j+1},1]$ and 0 for $x \in [0,t_{j+1}]$.
We can check that the conditions in the induction hypothesis hold for $q_{j+1}$ and $t_{j+1}$ as well.  This completes the induction proof.
\end{proof}

\subsection{Explicit Methods for the $(J,1)$-Case}
\label{sec:generate}

For $K=1$ with thresholds denoted by $t_j = \tau_{j,1}$ for $j\in[J]$, the
proof of Lemma~\ref{lemma:slack} gives a method to generate the dual variables $q_j$'s and thresholds $t_j$'s, which we describe below.


\begin{table}[H]
\begin{tabular}{|p{\textwidth}|}
\hline
\textsc{$J$-ThresholdsGenerator} \\
Set $\theta_1 := 1$ and $t_1 := e^{-\theta_1}$. Let $q_1$ be a function defined on $[0,1]$ such that $q_1(x) = 0$ for $x\in[0, t_1)$ and $q_1(x) = 1+\ln x$ for $x\in[t_1, 1]$.\\

For $j=1,2,\cdots, J-1$
\begin{enumerate}
\item[]
Set $\theta_{j+1} := 1 + \int_{t_j}^{1} \frac{q_j(y)}{y} d y$ and $t_{j+1} := e^{-\theta_{j+1}} \in (0, t_j)$.

\item[]
Let $q_{j+1}$ be a function defined on $[0,1]$ such that

\centering
$
q_{j+1}(x) =
\begin{cases}
0, & 0 \leq x < t_{j+1} \\

1 + \ln x + \int_{t_j}^{1} \frac{q_j(y)}{y} d y, & t_{j+1} \leq x < t_{j} \\

1 + \ln x + \int_{x}^{1} \frac{q_j(y)}{y} d y, & t_{j} \leq x \leq 1.
\end{cases}
$
\end{enumerate} \\
\hline
\end{tabular}
\end{table}

One can see that the $q_j(x)$'s are polynomials in $\ln x$ with rational coefficients.  Hence
we can describe the algorithm by maintaining the
rational coefficients of the polynomials.  The modified method is described
as follows, and it can be seen that generating the first $J$ thresholds takes $O(J^3)$ time.
\vspace{-2pt}

\begin{table}[h]
\begin{tabular}{|p{\textwidth}|}
\hline
\textsc{$\theta$-Generator} \\
For integer $n\geq 1$, let $0^n$ be the zero vector with $n$ coordinates. Let $\theta_0 := 0$. \\

For positive integers $j$ and $k$, let $c_{j,k} \in \R^{J+1}$ be a vector (corresponding to $q_j$ in interval $[t_k, t_{k-1}]$ where $t_0 := 1$). Denote by $c_{j,k}(i)$ the $i$-th coordinate of $c_{j,k}$ for $i\in[J+1]$. \\

Set $c_{1,1} := (1, 1, 0^{J-1})$ and $\theta_1 := 1$. \\

For $j=1,2,3,\ldots$ \\

\hspace{10pt} Let $\alpha \in \R$ be an auxiliary variable with initial value $\alpha := 0$. \\

\hspace{10pt} For $k=1,2,\ldots, j$ \\

\hspace{20pt} Let $d\in \R^{J}$ be such that $d(i) = -\frac{c_{j,k}(i)}{i}$ for each $i\in[J]$. \\

\hspace{20pt} If $k>1$, then set $\alpha := \alpha + \sum_{i=1}^{J+1} \frac{c_{j,k-1}(i)}{i} [(-\theta_{k-2})^i-(-\theta_{k-1})^i]$.

\hspace{20pt} Set $c_{j+1, k} := (1, 1, 0^{J-1})
+ (\sum_{i=1}^{J+1} \frac{c_{j,k}(i)}{i} (-\theta_{k-1})^i, d )
+ (\alpha, 0^J )$. \\

\hspace{10pt} Set $\alpha := \alpha + \sum_{i=1}^{J+1} \frac{c_{j,j}(i)}{i} [(-\theta_{j-1})^i-(-\theta_{j})^i]$.

\hspace{10pt} Set $c_{j+1,j+1} := (1,1,0^{J-1})
+ (\alpha, 0^J )$. \\

\hspace{10pt} Set $\theta_{j+1} := c_{j+1,j+1}(1)$. \\
\hline
\end{tabular}
\end{table}

\section{Primal-Dual Method for General $(J,K)$-case}
\label{sec:k-best}
In Section~\ref{sec:optalg} we have shown the optimality of the $J$-Threshold Algorithm. We apply our primal-dual method to the general $(J,K)$-case following a similar framework. After proving Theorem~\ref{th:main0} by construction, we can directly obtain Theorem~\ref{th:12_22} as a special case.

For $k\in[K]$ and $x\in[0,1]$, define $\alpha_k(x) := \txts\sum_{\ell=k}^{K} \binom{\ell-1}{k-1} (1-x)^{\ell-k} x^{k-1}$. The dual continuous LP for $\CP(J,K)$ is as follows.
\begin{align*}
\CD(J,K) \qquad \min \qquad & \txts\sum_{k=1}^{K} \int_{0}^{1} q_{J|k}(x) d x \\
\text{s.t.} \qquad & q_{1|k} (x) + \txts\frac{1}{x} \int_{x}^{1} \sum_{\ell=1}^{K} q_{1|\ell}(y) d y \geq \alpha_k(x), \qquad \forall x\in[0,1], k\in[K] \\
& q_{j|k}(x) + \txts\frac{1}{x} \int_{x}^{1} \sum_{\ell=1}^{K} [q_{(j)|\ell}(y) - q_{(j-1)|\ell}(y)] d y \geq \alpha_k(x) \\
& \qquad\qquad \forall x\in[0,1], k\in[K], 1 < j \leq J \\
& q_{j|k}(x) \geq 0, \qquad \forall x\in[0,1], k\in[K], j\in[J].
\end{align*}


For $\CP(J,K)$ and $\CD(J,K)$, we say a constraint is the $(j,k)$-th constraint if $p_{j|k}$ or $q_{j|k}$ is the concerned function in the constraint. For instance, the $(j,k)$-th constraint in the dual with $1 < j \leq J$ is $q_{j|k}(x) + \txts\frac{1}{x} \int_{x}^{1} \sum_{\ell=1}^{K} [q_{(j)|\ell}(y) - q_{(j-1)|\ell}(y)] d y \geq \alpha_k(x)$, $\forall x\in[0,1]$. We still have weak duality and the following complementary slackness conditions.

\begin{fact} [Complementary Slackness Conditions] \label{fact:css_jk}
Let $p=(p_{j|k})_{j\in[J], k\in[K]}$ and $q=(q_{j|k})_{j\in[J],k\in[K]}$ be feasible solutions of $\CP(J,K)$ and $\CD(J,K)$, respectively. Then $p$ and $q$ are primal and dual optimal, respectively, if they satisfy the following conditions $\forall x\in[0,1], k\in[K]$:
\begin{align*}
& \txts \left( p_{j|k}(x) + \txts\int_{0}^{x} \frac{1}{y} \sum_{\ell=1}^{K} [p_{j|\ell}(y) - p_{(j+1)|\ell} (y)] d y \right) q_{j|k}(x) = 0, \qquad 1 \leq j < J \\
& \txts \left( p_{J|k}(x) + \txts\int_{0}^{x} \frac{1}{y} \sum_{\ell=1}^{K} p_{J|\ell}(y) d y - 1 \right) q_{J|k}(x) = 0 \\
& \txts \left( q_{j|k}(x) + \txts\frac{1}{x} \int_{x}^{1} \sum_{\ell=1}^{K} [q_{(j)|\ell}(y) - q_{(j-1)|\ell}(y)] d y - \alpha_k(x) \right) p_{j|k}(x) = 0, \qquad 1 < j \leq J \\
& \txts \left( q_{1|k} (x) + \txts\frac{1}{x} \int_{x}^{1} \sum_{\ell=1}^{K} q_{1|\ell}(y) d y - \alpha^k(x) \right) p_{1|k}(x) = 0. 
\end{align*}
\end{fact}

\noindent \textbf{Primal-Dual Method.} We start from a primal feasible solution $p$ corresponding
to a $(J,K)$-Threshold Algorithm, whose thresholds are to be determined.
We can determine the values of the thresholds one by one in order
to construct a dual $q$ such that complementary slackness conditions
hold, which implies that with those found thresholds the $(J,K)$-Threshold
Algorithm is optimal.

\noindent \textbf{(1) Forming Feasible Primal Solution $p$.}
Suppose $p$ is the (feasible) primal correspond to the
$(J,K)$-Threshold Algorithm with $JK$ thresholds 
$\tau_{j,k}$ such that $0< \tau_{J,k} \leq \tau_{J-1,k} \leq \cdots \leq \tau_{1,k} \leq 1$ for $k\in [K]$ and $0 <\tau_{j,1} \leq \tau_{j,2}\leq \cdots \leq \tau_{j,K} \leq 1$ for $j\in[J]$.
Suppose $E_x^j$ is the event that the item at $x$ is selected
by using quota $Q_j$ (where quotas with larger $j$'s are used first),
$V_x^k$ is the event that $x$ is a $k$-potential, 
and $Z_x^j$ is the event that at time $x$, quota $Q_j$ has
already been used (and so have the quotas with indices larger than $j$).
For notational convenience, $Z_x^{J+1}$ is the whole sample space, i.e.,
an always true event.

For each $j \in [J]$, consider the conditional probability
$\Pr(E_x^j |Z_x^{j+1}\wedge\overline{Z_x^j}\wedge V_x^k)$
of the event that item at $x$ is selected
by using quota $Q_j$, given that $x$ is a $k$-potential and quota $Q_j$ is the next available quota at time $x$.  By definition of the Threshold
Algorithm, this conditional probability is 0 if $x < \tau_{j,k}$ and is 1 if $x \geq \tau_{j,k}$.  Hence, we have the following.

\begin{align*}
\txts \Pr(E_x^j|Z_x^{j+1}\wedge\overline{Z_x^j}\wedge V_x^k)  = \txts\frac{\Pr(E_x^j|V_x^k)}{\Pr(Z_x^{j+1}\wedge\overline{Z_x^j} | V_x^k)} = \txts \frac{p_{j|k}(x)}{\Pr(Z_x^{j+1}\wedge\overline{Z_x^j} | V_x^k)} 
 = 
\begin{cases}
 0, & 0 \leq x < \tau_{j,k} \\
 1, & \tau_{j,k} \leq x \leq 1,
\end{cases}
\end{align*}
where from independence of $V_x^k$ and $Z_x^j$ (Lemma~\ref{lemma:ind}), and Lemma~\ref{lemma:k-trick}, we have:
\begin{align*}
\Pr(Z_x^{j+1}\wedge\overline{Z_x^j} | V_x^k) 
= \begin{cases}
\txts\int_{0}^{x} \frac{1}{y} \sum_{\ell=1}^{K} [p_{(j+1)|\ell}(y) - p_{j|\ell} (y)] d y, & 1 \leq j < J \\
1 - \txts\int_{0}^{x} \frac{1}{y} \sum_{\ell=1}^{K} p_{J|\ell}(y) d y, & j = J.
\end{cases}
\end{align*}

This implies that in the primal $\CP(J,K)$,
the $(j,k)$-th constraint is equality
in the range $[\tau_{j,k}, 1]$, but might be strict inequality
in the range $[0, \tau_{j,k})$ (hence forcing $q_{j|k}$ to be 0);
the function $p_{j|k}$ is zero in the range $[0,\tau_{j,k})$, but
might be strictly positive in the range $[\tau_{j,k}, 1]$ (hence
forcing equality for the $(j,k)$-th constraint in dual).

\noindent \textbf{(2) Finding Feasible Dual $q$ to Satisfy
Complementary Slackness.}  To ensure that a dual solution $q$
satisfies complementary slackness together with the above primal $p$.
We require the following for each $j \in [J]$ and $k\in[K]$, where for
notational convenience we write $q_{0|k} \equiv 0$ for all $k\in[K]$.

\begin{equation} \label{eq:slack_jk}
\begin{cases}
q_{j|k}(x) = 0, & x \in [0, \tau_{j,k}] ;\\
q_{j|k}(x) + \txts\frac{1}{x} \int_{x}^{1} \sum_{\ell=1}^{K} [q_{j|\ell}(y) - q_{(j-1)|\ell}(y)] d y = \alpha_k(x), & x \in  [\tau_{j,k}, 1].
\end{cases}
\end{equation}

Here the extra condition $q_{j|k}(\tau_{j,k}) = 0$ ensures that as long as~\eqref{eq:slack_jk} is satisfied by some
non-negative $q_{j|k}$, the $(j,k)$-th constraint in $\CD(J)$ is also automatically
satisfied.
Proof for this indication is not straightforward and requires stronger conditions for the dual functions, which we provide along the way we prove Theorem~\ref{th:main0}.
From the recursive equations~\eqref{eq:slack_jk}, the thresholds $\tau_{j,k}$'s and functions $q_{j|k}$'s found for $\CD(J,K)$ can be used to extend to the solution for $\CD(J+1,K)$.

\noindent\textbf{Objective Value.} The objective value of $\CD(J,K)$ is $\int_{0}^{1}\sum_{k=1}^{K} q_{J|k}(x) d x = \int_{\tau_{J,1}}^{1}\sum_{k=1}^{K} q_{J|k}(x) d x$. From equations~\eqref{eq:slack_jk} we have for $j\in[J]$,
\begin{align*}
\txts\int_{\tau_{j,1}}^{1} \sum_{k=1}^{K} q_{j|k}(x) d x - \int_{\tau_{j-1,1}}^{1} \sum_{k=1}^{K} q_{(j-1)|k}(x) d x = \tau_{j,1} \alpha_1(\tau_{j,1}),
\end{align*}
where $\tau_{0,1} = 1$. This together with $\alpha_1(x) = \frac{1-(1-x)^K}{x}$ implies that the objective value of $\CD(J,K)$ is
\begin{align*}
\txts\int_{\tau_{J,1}}^{1}\sum_{k=1}^{K} q_{J|k}(x) d x = \sum_{j=1}^{J} \tau_{j,1} \alpha_1(\tau_{j,1})
= J - \sum_{j=1}^{J} (1-\tau_{j,1})^K.
\end{align*}

To prove Theorem~\ref{th:main0}, it suffices to show the existence of dual functions as required in~\eqref{eq:slack_jk}. Before showing this result we first give two useful observations.

\begin{lemma} \label{lemma:tech1}
Let $b>0$ and $c$ be real numbers, $N$ be a positive integer, and $g(x)$ and $\gamma(x)$ be two functions of $x$ continuous in $(0,b]$. Then, the equation $f(x) + \frac{N}{x} \int_{x}^{b} [f(y) - g(y)] d y + \frac{c}{x} = \gamma(x)$ with respect to $f$ has a continuous solution in $(0,b]$, which can be expressed as
\begin{align*}
f(x) = x^{N-1} \left[\txts\frac{b \gamma(b)-c}{b^N} - \int_{x}^{b} \frac{(y \gamma(y))'}{y^N} d y + N \int_{x}^{b} \frac{g(y)}{y^N} d y \right].
\end{align*}
In particular, if we replace $c$ with $\widehat{c}$ and $g$ with $\widehat{g}$ such that $\widehat{c}<c$ and $\widehat{g}>g$, then the resulting solution $\widehat{f}$ satisfies $\widehat{f}>f$.
\end{lemma}

\begin{proof}
Substituting $x$ with $b$ into the equation we get $f(b) = \gamma(b) - \frac{c}{b}$. Taking derivatives on both sides of $x f(x) + N \int_{x}^{b} [f(y) - g(y)] d y + c = x \gamma(x)$ we get
\begin{align*}
x f'(x) - (N-1) f(x) + N g(x) = (x \gamma(x))'.
\end{align*}
It suffices to show there exists a continuous $f(x)$ satisfying the above equation with condition $f(b) = \gamma(b) - \frac{c}{b}$. The above equation is equivalent to $\left( \frac{f(x)}{x^{N-1}} \right)' = \frac{(x \gamma(x))' - N g(x)}{x^N}$, which gives $f(x) = x^{N-1} [c_0 - \int_{x}^{b} \frac{(y \gamma(y))' - N g(y)}{y^N} d y ]$ for some constant $c_0$. Using the initial condition we get $c_0 = \frac{b \gamma(b)-c}{b^N}$ and hence a continuous function $f(x) = x^{N-1} [\frac{b \gamma(b)-c}{b^N} - \int_{x}^{b} \frac{(y \gamma(y))'}{y^N} d y + N \int_{x}^{b} \frac{g(y)}{y^N} d y]$.
\end{proof}

\begin{lemma} \label{lemma:tech2}
Let $K>1$ be an integer. For $k\in[K]$ and $x\in[0,1]$, define $\alpha_k(x) := x^{k-1} \sum_{\ell=k}^{K} \binom{\ell-1}{k-1} (1-x)^{\ell-k}$. Then for all $x\in(0,1)$ we have the following.
\begin{compactitem}
\item[(a)]
$(x\alpha_k(x))' > 0$ for $1\leq k \leq K$;

\item[(b)]
$\alpha_{k}(x) > \alpha_{k+1}(x)$ for $1 \leq k < K$.
\end{compactitem}
\end{lemma}

\begin{proof}

Observe that $\alpha_K(x) = x^{K-1}$. Then $(x \alpha_K(x))' = K x^{K-1} > 0$. In what follows, we first show that $(x \alpha_k(x))' > 0$ if and only if $\alpha_{k}(x) > \alpha_{k+1}(x)$ for $1 \leq k < K$. Then we prove $\alpha_{k}(x) > \alpha_{k+1}(x)$ by induction starting with $k = K-1$.

Let $1 \leq k < K$. Define $\beta_k(x) := \frac{\alpha_k(x)}{x^{k-1}}$; note $\beta_K(x) = 1$. We show $(x^{k} \beta_k(x))' > 0$ if and only if $\beta_{k}(x) > x \beta_{k+1}(x)$. By definition we have
\begin{align*}
\beta'_{k}(x) & = \left( \txts\sum_{\ell=k+1}^{K} \binom{\ell-1}{k-1} (1-x)^{\ell-k} \right)'
= - \txts\sum_{\ell=k+1}^{K} \binom{\ell-1}{k-1} (\ell-k) (1-x)^{\ell-k-1} \\
& = - k \txts\sum_{\ell=k+1}^{K} \binom{\ell-1}{k} (1-x)^{\ell-k-1}
= - k \beta_{k+1}(x).
\end{align*}
It follows that
\begin{align*}
(x^{k} \beta_k(x))' > 0
\Longleftrightarrow k \beta_k(x) + x \beta'_k(x) > 0
\Longleftrightarrow \beta_{k}(x) > x \beta_{k+1}(x).
\end{align*}

Next we prove $\beta_{k}(x) > x \beta_{k+1}(x)$ by backward induction
starting at $k=K-1$. Note that
we have $\beta_{K-1}(x) = 1 + (K-1)(1-x) > x \beta_K(x)$. Suppose $1\leq k < K-1$ and the inequality holds for $k+1$, i.e., $\beta_{k+1}(x) - x \beta_{k+2}(x) > 0$.

Define $\lambda_k(x) := \beta_{k}(x) - x \beta_{k+1}(x)$. Note that $\beta_k(1) = \alpha_k(1) = 1$ and hence $\lambda_k(1) = 0$ for all $1\leq k < K$. Moreover,
 we have
\begin{align*}
\lambda'_k(x) & = \beta'_k(x) - \beta_{k+1}(x) - x \beta'_{k+1}(x) \\
& = - k \beta_{k+1}(x) - \beta_{k+1}(x) + (k+1) x \beta_{k+2}(x) \\
& = - (k+1) (\beta_{k+1}(x) - x \beta_{k+2}(x)) < 0,
\end{align*}
where the last inequality follows from the induction hypothesis. Therefore, we have $\lambda_k(x) > \lambda_k(1) = 0$, which implies $\beta_{k}(x) > x \beta_{k+1}(x)$.
\end{proof}

\begin{lemma} [Existence of Feasible Dual Satisfying Complementary Slackness] \label{lemma:slack_jk}
There is a procedure to find appropriate thresholds $(\tau_{j,k})_{j\in [J], k\in[K]}$ and a collection $(q_{j|k})_{j\in[J],k\in[K]}$ of non-negative functions that satisfy~\eqref{eq:slack_jk}.
\end{lemma}

\begin{proof}
We show the result by induction; our induction proof gives a method to generate the thresholds $\tau_{j,k}$'s and the functions $q_{j|k}$'s. For convenience we denote $q_{0|k}(x) \equiv 0$, and set $\tau_{0,k} := 1$ and $\tau_{j,K+1} := 1$ for all $k\in[K+1]$ and $j\in[J]$. 
Also, define $r_{j|k}(x) := \sum_{\ell=1}^{k} q_{j|\ell}(x)$ and $\gamma_k(x) = \sum_{\ell=1}^{k} \alpha_\ell(x)$ for all $j\in[J]$ and $k\in[K]$. Observe that $\gamma_K(x) \equiv K$ and $\gamma_k(0) = K$ for $k\in[K]$. The induction process is over $j\in[J]$. For each $j$, since each constraint involves the functions $q_{j|k}$ for all $k\in[K]$, we do not find $q_{j|k+1}$ on the whole interval $[0,1]$ before going to $q_{j|k}$; instead, we consider the intervals $[\tau_{j,k}, \tau_{j,k+1}]$ one by one and study the behavior of all functions $(q_{j|\ell})_{\ell\in[K]}$ within each interval.

\noindent\textbf{Base Case ($j=1$).} First consider the base case with $j=1$. To find thresholds $\tau_{1,k}$'s and non-negative functions $q_{1,k}$'s for $k\in[K]$ satisfying~\eqref{eq:slack_jk}, we use another induction procedure on $k$. Suppose $k=K$ and $x\in[\tau_{1,K}, 1]$. Summing up the equalities $q_{1|k}(x) + \txts\frac{1}{x} \int_{x}^{1} \sum_{\ell=1}^{K} q_{(j)|\ell}(y) d y = \alpha_k(x)$ over $k$ we get $r_{1|K}(x) + \frac{K}{x} \int_{x}^{1} r_{1|K}(y) d y = K$. By Lemma~\ref{lemma:tech1} we have $r_{1|K}(x) = \frac{K^2}{K-1} x^{K-1} - \frac{K}{K-1}$. Then it follows that $q_{1|k}(x) = \alpha_k(x) + \frac{K}{K-1} x^{K-1} - \frac{K}{K-1}$. From the equation $q_{1|K}(x) + \frac{1}{x} \int_{x}^{1} r_{1|K}(y) d y = \alpha_K(x)$ and Lemma~\ref{lemma:tech2} we have $r_{1|K}(x) \geq K q_{1|K}(x)$ and hence $q'_{1|K}(x) = \frac{(x\alpha_K(x))' + r_{1|K}(x) - q_{1|K}(x)}{x} \geq 0$. Recall $\alpha_K(x) = x^{K-1}$. Setting $q_{1|K}(\tau_{1,K}) = 0$ yields $\tau_{1,K} = \sqrt[K-1]{\frac{K}{2K-1}}$. Then $q_{1|K}(x) \geq 0$ for $x\in[\tau_{1,K}, 1]$. By Lemma~\ref{lemma:tech2}(b) we have $q_{1|k}(x) > q_{1|K}(x) \geq 0$ for $1\leq k \leq K-1$ and $x\in[\tau_{1,K}, 1)$.

Suppose for some $k \leq K-1$ we have found $\tau_{1,k+1} < \tau_{1,k+2}$ and constructed continuous $q_{1|\ell}$ for all $\ell\in [K]$ in $[\tau_{1,k+1}, 1]$, where $q_{1|\ell}$ can be positive only in $[\tau_{1,\ell}, 1]$ for $\ell>k$ and $q_{1|\ell}(\tau_{1,k+1}) > q_{1|\ell+1}(\tau_{1,k+1})$ for $\ell<k+1$. Consider the case $x < \tau_{1,k+1}$. Set $d := \int_{\tau_{1,k+1}}^{1} r_{1|K}(x) d x$. Then by $q_{1|k+1}(\tau_{1,k+1}) + \frac{d}{\tau_{1,k+1}} = \alpha_{k+1}(\tau_{1,k+1})$ and $q_{1|k+1}(\tau_{1,k+1}) = 0$ we get $d = \tau_{1,k+1} \alpha_{k+1}(\tau_{1,k+1})$.
If there exist $\tau_{1,k}<\tau_{1,k+1}$ and $q_{1|\ell}$ for $\ell\in[k]$ that satisfy~\eqref{eq:slack_jk}, then we must have the following:
\begin{align*}
& q_{1|\ell}(x) + \txts\frac{1}{x} \int_{x}^{\tau_{1,k+1}} r_{1|\ell}(y) d y + \frac{d}{x} = \alpha_\ell(x), \qquad \ell\in[k] \\
& r_{1|k}(x) + \txts\frac{k}{x} \int_{x}^{\tau_{1,k+1}} r_{1|k}(y) d y + \frac{kd}{x} = \gamma_k(x).
\end{align*}
By Lemma~\ref{lemma:tech1} and the above equations we must have
\begin{align}
& r_{1|k}(x) = x^{k-1} [\txts\frac{\gamma_k(\tau_{1,k+1}) - k\alpha_{k+1}(\tau_{1,k+1})}{\tau^{k-1}_{1,k+1}} - \int_{x}^{\tau_{1,k+1}}\frac{(y \gamma_k(y))'}{y^k}d y] \label{eq:r1} \\
& q_{1|\ell}(x) = \txts\frac{r_{1|k}(x) - \gamma_k(x)}{k} + \alpha_\ell(x), \qquad \ell\in[k]. \label{eq:q1}
\end{align}
From~\eqref{eq:r1} and~\eqref{eq:q1} the function $q_{1,k}$ must agree on $[\tau_{1,k},\tau_{1,k+1}]$ with $q(x)$, where $q(x)$ is given by $q(x) =\frac{r(x) - \gamma_k(x)}{k} + \alpha_k(x)$ and $r(x)$ is given by $r(x) = x^{k-1} [\txts\frac{\gamma_k(\tau_{1,k+1}) - k\alpha_{k+1}(\tau_{1,k+1})}{\tau^{k-1}_{1,k+1}} - \int_{x}^{\tau_{1,k+1}}\frac{(y \gamma_k(y))'}{y^k}d y]$.
Observe that $q(\tau_{1,k+1}) + \frac{d}{\tau_{1,k+1}} = \alpha_k(\tau_{1,k+1})$ and hence $q(\tau_{1,k+1}) = q_{1|k}(\tau_{1,k+1}) > q_{1|k+1}(\tau_{1,k+1}) = 0$.
On the other hand, we have $r(x) > k q(x)$ and thus $q'(x) > 0$. In~\eqref{eq:r1}, the term $(y\gamma_k(y))'$ is a positive polynomial in $y$ and takes value $K$ when $y=0$; thus $r(x)$ is negative when $x$ is close to $0$. It follows that $q(x)$ is also negative when $x$ is close to $0$ and hence the equation $q(x)=0$ has a unique solution in $(0,\tau_{1,k+1})$. Let $\tau_{1,k}\in(0,\tau_{1,k+1})$ be such that $q(\tau_{1,k})=0$. Then, we set $q_{1|\ell} := \frac{r(x) - \gamma_k(x)}{k} + \alpha_\ell(x)$ for $\ell\in[k]$ and $x\in[\tau_{1,k},\tau_{1,k+1}]$. By Lemma~\ref{lemma:tech2}(b) for $\ell < k$ we have $q_{1|\ell}(x)>q_{1|k}(x)$ and in particular $q_{1|\ell}(\tau_{1,k})>0$. It can be easily checked that the functions $q_{1|\ell}$'s are continuous in $[\tau_{1,k}, 1]$.
The base case with $j=1$ is completed.

\noindent \textbf{Inductive Step $(j,k)$.} Let $j \geq 2$ and $k \leq K$.
We state the induction hypothesis; the base case we proved above corresponds
to $j=2$ and $k = K$.
For each of the conditions, we also state (in parentheses) what we need to prove
for the inductive step.

\noindent\textbf{Induction Hypothesis.} Let $j \geq 2$ and $k \leq K$. Suppose we have constructed the following thresholds and functions satisfying~\eqref{eq:slack_jk}: (1) Thresholds $\tau_{i,\ell}$ for all $i<j$, $\ell\in[K]$ and $\tau_{j,\ell}$ for $\ell>k$, where $\tau_{i,\ell} < \tau_{i,\ell+1}$ and $\tau_{i,\ell}<\tau_{i-1,\ell}$ for appropriate $i,\ell$. (2) Functions $q_{i,\ell}$ for all $i<j$, $\ell\in[K]$ and $q_{j,\ell}$ for $\ell>k$. 
Moreover, the following holds.
\begin{compactitem}
\item[1.]
The functions $q_{i|\ell}$'s for all $i<j$, $\ell\in[K]$ are non-negative and continuous in $[0,1]$. The functions $q_{j|\ell}$'s for $\ell>k$ are non-negative and continuous in $[0,1]$.  
(We shall find $\tau_{j,k} < \min\{\tau_{j,k+1}, \tau_{j-1,k}\}$
and complete the definition of the function $q_{j|k}$.)

\item[2.]
In interval $[\tau_{j,k+1},1]$, the functions $q_{j|\ell}$'s for $\ell\in[k]$ are determined and they are non-negative and continuous . Moreover, they satisfy equation~\eqref{eq:q_rec}.
(We shall extend these functions to the range $[\tau_{j,k}, 1]$.)

\item[3.]
It holds that $r_{j|K}(x) > r_{j-1|K}(x)$ for $x\in(\tau_{j,k+1}, 1)$, and $r_{i|K}(x) > r_{i-1|K}(x)$ for $1\leq i < j$ and $x\in(\tau_{i,1},1)$.
(We shall show that $r_{j|K}(x) > r_{j-1|K}(x)$ for $x \in (\tau_{j,k},1)$.)


\item[4.]
In interval $[\tau_{j,\ell}, 1]$, where $k<\ell\leq K$, it holds that $q_{j|m}(x)>q_{j|m+1}(x)$ for $1\leq m < \ell$.
(We shall show that in the interval $[\tau_{j,k},1]$,
for $1 \leq m < k$,  $q_{j|m}(x)>q_{j|m+1}(x)$.)
\end{compactitem}


\noindent\textbf{Dual Feasibility.} Before we prove the
inductive step, we
first show that the above hypothesis implies dual feasibility.

\begin{lemma}[Dual Feasibility] \label{lemma:dual_feas}
Consider dual functions and thresholds that satisfy (\ref{eq:slack_jk}),
and in particular suppose for some $i > 0$, we have $r_{i|K}(x) > r_{i-1|K}(x)$
on $(\tau_{i,1}, 1)$ (Condition 3).
Then, for all $k \in [K]$, the $(i,k)$-th constraint in $\CD(J,K)$ is satisfied;
moreover, strict inequality holds for $x\in[0,\tau_{i,k})$.
\end{lemma}

\begin{proof}
From (\ref{eq:slack_jk}), the $(i,k)$-th constraint is 
equality for $x\in[\tau_{i,k}, 1]$. Since~\eqref{eq:slack_jk} holds at $x=\tau_{i,k}$, we have $\frac{1}{\tau_{i,k}}\int_{\tau_{i,k}}^{1}[r_{i|K}(y)-r_{i-1|K}(y)] d y = \alpha_k(\tau_{i,k})$. Suppose $x\in[0,\tau_{i,k})$, then by Lemma~\ref{lemma:tech2}(a) we have $\tau_{i,k} \alpha_k(\tau_{i,k}) > x \alpha_k(x)$. Observe that condition 3 implies that $r_{i|K}(x) \geq r_{i-1|K}(x)$,
as the two functions are equal outside $(\tau_{i,1},1)$. It follows that
$ q_{i|k}(x) + \txts\frac{1}{x}\int_{x}^{1}[r_{i|K}(y)-r_{i-1|K}(y)] d y
\geq \frac{1}{x}\int_{\tau_{i,k}}^{1}[r_{i|K}(y)-r_{i-1|K}(y)] d y
 = \txts\frac{\tau_{i,k} \alpha_k(\tau_{i,k})}{x} > \alpha_k(x)$.
\end{proof}



\noindent \textbf{Conditions 1 and 2.}
Now we wish to determine the threshold $\tau_{j,k}$ and define functions $q_{j,\ell}$'s for $\ell\in[k]$ in $[\tau_{j,k},\tau_{j,k+1}]$. Define $d_j := \int_{\tau_{j,k+1}}^{1} [r_{j|K}(y) - r_{j-1|K}(y)]d y$. If such $\tau_{j,k}$ and $q_{j,\ell}$ for $\ell\in[k]$ exist, then by~\eqref{eq:slack_jk} we must have the following for $\ell\in[k]$ and $x\in[\tau_{j,k},\tau_{j,k+1}]$:
\begin{align}
q_{j|\ell}(x) + \txts\frac{1}{x} \int_{x}^{\tau_{j,k+1}} [r_{j|K}(y) - r_{j-1|K}(y)] d y + \frac{d_j}{x} = \alpha_\ell(x). \label{eq:q_rec}
\end{align}
Since $q_{j|k+1}$ satisfies~\eqref{eq:slack_jk} at point $\tau_{j,k+1}$, we have $d_j = \tau_{j,k+1} \alpha_{k+1}(\tau_{j,k+1})$. 
Summing up the above equation over $\ell\in[k]$,
and observing that $r_{j|K}(x) = r_{j|k}(x)$ for $x \leq \tau_{j,k+1}$, we get
\begin{align}
r_{j|k}(x) + \txts\frac{k}{x} \int_{x}^{\tau_{j,k+1}} [r_{j|k}(y) - r_{j-1|K}(y)] d y + \frac{k \tau_{j,k+1} \alpha_{k+1}(\tau_{j,k+1})}{x} = \gamma_k(x). \label{eq:r_rec}
\end{align}

By Lemma~\ref{lemma:tech1}, we can conclude that
$r_{j|k}(x)$ must agree on $[\tau_{j,k}, \tau_{j,k+1}]$
with the following function $r(x)$:
\begin{align}
r(x) = x^{k-1} \txts\left[\txts\frac{\gamma_k(\tau_{j,k+1})-k \alpha_{k+1}(\tau_{j,k+1})}{\tau_{j,k+1}^{k-1}} - \int_{x}^{\tau_{j,k+1}} \frac{(y \gamma_k(y))' - k r_{j-1|K}(y)}{y^k} d y \right]. \label{eq:r_from_rec}
\end{align}

Comparing equations~\eqref{eq:q_rec} and~\eqref{eq:r_rec}, we
conclude that for $\ell\in[k]$, for $x \in [\tau_{j,k}, \tau_{j,k+1}]$,
\begin{align}
q_{j|\ell}(x) = \txts\frac{r_{j|k}(x) - \gamma_k(x)}{k} + \alpha_\ell(x). \label{eq:q_from_r}
\end{align}

%
%

In particular, $q_{j|k}$ must agree on $[\tau_{j,k},\tau_{j,k+1}]$ with the function 
$q(x) = \frac{r(x) - \gamma_k(x)}{k} + \alpha_k(x)$.

%


\noindent Let $\widehat{\tau} := \min\{\tau_{j-1,k}, \tau_{j,k+1}\}$. We
wish to extend the $q_{j|l}$'s to $[\widehat{\tau}, \tau_{j,k+1}]$ first.
For  $\widehat{\tau} = \tau_{j,k+1} < \tau_{j-1,k}$,
this is done and by induction hypothesis,
we have $q(\widehat{\tau}) = q_{j|k}(\tau_{j,k+1}) > q_{j|k+1}(\tau_{j,k+1}) = 0$.
Next, we assume $\widehat{\tau} = \tau_{j-1,k} \leq \tau_{j,k+1}$.
Suppose $x \in[\tau_{j-1,k},\tau_{j,k+1}]$. Note that $\tau_{j,k+1} < \tau_{j-1,k+1}$. Thus we have
\begin{align*}
r_{j-1|k}(x) + \txts\frac{k}{x} \int_{x}^{\tau_{j,k+1}} [r_{j-1|k}(y) - r_{j-2|K}(y)] d y + \frac{k \tau_{j-1,k+1} \alpha_{k+1}(\tau_{j-1,k+1})+kc}{x} = \gamma_k(x),
\end{align*}
where $c = \int_{\tau_{j,k+1}}^{\tau_{j-1,k+1}} [r_{j-1|K}(y) - r_{j-2|K}(y)] d y > 0$, because by the hypothesis $r_{j-1|K}(y) > r_{j-2|K}(y)$
for $y \geq \tau_{j-1,k} > \tau_{j-1,1}$.


Comparing the above equation with~\eqref{eq:r_rec} and using Lemma~\ref{lemma:tech1} and Lemma~\ref{lemma:tech2}(a) we have $r(x) > r_{j-1|k}(x)$ for all $x\in[\tau_{j-1,k},\tau_{j,k+1}]$. Also, we have $q(x) = \frac{r(x) - \gamma_k(x)}{k} + \alpha_k(x) > \frac{r_{j-1|k}(x) - \gamma_k(x)}{k} + \alpha_k(x) = q_{j-1|k}(x) \geq 0$. Now, for $x\in[\tau_{j-1,k},\tau_{j,k+1}]$, we set $q_{j|\ell}(x) := \frac{r(x) - \gamma_k(x)}{k} + \alpha_\ell(x)$ for $\ell\in[k]$. Then
\begin{align}
r_{j|K}(x) = r_{j|k}(x) = r(x) > r_{j-1|K}(x). \label{eq:r_ineq}
\end{align}
Also $q_{j|\ell}(x)>0$ for $\ell\in[k]$ by Lemma~\ref{lemma:tech2}(b). It can be easily checked that the functions $q_{j|\ell}$'s are continuous. In particular, we also have $q(\widehat{\tau}) > 0$.

Hence, we have show that for $x \in [\widehat{\tau}, \tau_{j,k+1}]$,
$q(x) > 0$.
We next analyze the behavior of $r(x)$ and $q(x)$ as $x$ tends to $0$. In the definition of $r(x)$, the function $r_{j-1|K}(y)$ has value 0 when $y<\tau_{j-1,1}$; the term $(y\gamma_k(y))'>0$ is by definition a polynomial of $y$, which takes value $K$ when $y=0$. Then $\frac{(y \gamma_k(y))' - k r_{j-1|K}(y)}{y^k}>0$ is unbounded and hence $r(x)$ is negative when $x$ is close to $0$. Since $q(x) < k r(x)$, it holds that $q(x)$ is also negative when $x$ is close to $0$. Since $q$ is continuous in $(0,\widehat{\tau}]$, there exists $x\in(0,\widehat{\tau})$ such that $q(x) = 0$. Let $\tau_{j,k}$ be the largest value in $(0,\min\{\tau_{j,k+1}, \tau_{j-1,k} \})$ such that $q(\tau_{j,k}) = 0$. Then $q(x)>0$ for $x\in(\tau_{j,k}, \widehat{\tau}]$.

Now, for $x\in[\tau_{j,k}, \widehat{\tau}]$, we set $q_{j|\ell}(x) := \frac{r(x) - \gamma_k(x)}{k} + \alpha_\ell(x)$ for $\ell\in[k]$. 
It can be easily checked that the functions $q_{j|\ell}$'s are continuous.

\noindent \textbf{Condition 4.}
By Lemma~\ref{lemma:tech2}(b), we have $\alpha_m(x) > \alpha_{m+1}(x)$ for $x \in (0,1)$, and so we have $q_{j|m}(x)>q_{j|m+1}(x)$ for $1 \leq m < k$.  

\noindent \textbf{Condition 3.}
If $\tau_{j-1,k} \leq \tau_{j,k+1}$,
then from~\eqref{eq:r_ineq} we already have
$r_{j|K}(x) > r_{j-1|K}(x)$ for $x \in [\tau_{j-1,k}, \tau_{j,k+1}]$; we are left to show $r_{j|K}(x)>r_{j-1|K}(x)$ for $x\in(\tau_{j,k}, \widehat{\tau})$. If $\tau_{j-1,k} > \tau_{j,k+1}$, then we are left to show $r_{j|K}(x) >r_{j-1|K}(x)$ for $x\in(\tau_{j,k}, \widehat{\tau}]$. We next show that $r_{j|K}(x) = r_{j|k}(x)>r_{j-1|K}(x)$ for all $x\in \mathcal{I}$, where $\mathcal{I}$ is defined as
\begin{align*}
\mathcal{I} =
\begin{cases}
x\in(\tau_{j,k}, \widehat{\tau}), & \tau_{j-1,k} \leq \tau_{j,k+1} \\
x\in(\tau_{j,k}, \widehat{\tau}], & \tau_{j-1,k} > \tau_{j,k+1}.
\end{cases}
\end{align*}


%
%

Suppose on the contrary $r_{j|K}(x) \leq r_{j-1|K}(x)$ for some $x \in\mathcal{I}$. Since $r_{j|K}$ and $r_{j-1|K}$ are continuous, there exists $z\in\mathcal{I}$ with $z\geq x$ such that $0 < r_{j|K}(z) = r_{j-1|K}(z)$. Let $m<k$ be the integer such that $\tau_{j-1,m} \leq z < \tau_{j-1,m+1}$. Then $r_{j|k}(z) = r_{j-1|m}(z)$. By Lemma~\ref{lemma:dual_feas} we have
\begin{align*}
& q_{j-1|\ell}(z) + \txts\frac{1}{z}\int_{z}^{1} [r_{j-1|K}(y)-r_{j-2|K}(y)] d y = \alpha_\ell(z), \qquad 1\leq \ell \leq m \\
& q_{j-1|\ell}(z) + \txts\frac{1}{z}\int_{z}^{1} [r_{j-1|K}(y)-r_{j-2|K}(y)] d y > \alpha_\ell(z), \qquad m < \ell \leq k.
\end{align*}
Since at least one strict inequality holds for $\ell\in [k]$, we have
\begin{align*}
r_{j-1|k}(z) + \txts\frac{k}{z}\int_{z}^{1} [r_{j-1|K}(y)-r_{j-2|K}(y)] d y > \gamma_k(z).
\end{align*}
Then, for all $l \in [m]$, we get $q_{j-1|\ell}(z) < \frac{r_{j-1|k}(z) - \gamma_k(z)}{k} + \alpha_\ell(z) = \frac{r_{j|k}(z) - \gamma_k(z)}{k} + \alpha_\ell(z)$, because we have $r_{j-1|k}(z) = r_{j-1|m}(z) = r_{j|k}(z)$. From~\eqref{eq:q_from_r}, we have $q_{j-1|\ell}(z) < q_{j|\ell}(z)$ for $\ell\in[m]$. But this implies $r_{j-1|m}(z) < r_{j|m}(z) \leq r_{j|k}(z)$, a contradiction.

This completes the induction proof.


%
%
%

\end{proof}

{
\bibliographystyle{plain}
\bibliography{secretary}
}

\end{document}